\documentclass[USenglish,onecolumn]{article}

\usepackage[utf8]{inputenc}%(only for the pdftex engine)
\usepackage[big]{dgruyter}
 \usepackage{latexsym}
 \usepackage{amssymb,amsmath, bm,pgfplots,tikz,bbm}
 \usepackage{graphicx}
 \usepackage{multirow,float}
 \usepackage{caption}
 \usepackage{comment}
 \usepackage{makecell,centernot}
 \usepackage{color}
 \usepackage{bm}
 \usepackage{natbib}
 \usepackage{subcaption}
 \usepackage{algorithm}
 \usepackage{algpseudocode}
 \DeclareMathOperator\Cov{Cov}
 \DeclareMathOperator\Corr{Corr}

 \newcommand\iid{\stackrel{\rm i.i.d.}{\sim}}
 
 \renewcommand\r{\right}
 \renewcommand\l{\left}

 \newcommand\cZ{\mathcal{Z}}
 \newcommand\E{\mathbb{E}}
 
 \newcommand\V{\mathbb{V}}

 \newcommand\cX{\mathcal{X}}
 
  \newcommand\cS{\mathcal{S}}

 \newcommand\bx{\bm{x}}
 \newcommand\bX{\bm{X}}

 \newcommand\bT{\bm{T}}
 \newcommand\bt{\bm{t}}
 
 \newcommand\bZ{\bm{Z}}
 
 \newcommand\bF{\bm{F}}

 \newcommand\bone{\mathbf{1}}

 \usepackage{rotating}
 
 \usepackage{arydshln}
 \usepackage{threeparttable}
 \usepackage{pifont}
 \usepackage{xcolor}

\newtheorem{theorem}{Theorem}
\newtheorem{proposition}{Proposition}
\newtheorem{assumption}{Assumption}

\newtheorem{lemma}{Lemma}

\usetikzlibrary{decorations.markings}
\usetikzlibrary{decorations.pathmorphing}
\usetikzlibrary{shapes.geometric, arrows}
\usetikzlibrary{arrows,decorations.pathmorphing,backgrounds,positioning,fit,matrix}
\usetikzlibrary{shapes,decorations,arrows,calc,arrows.meta,fit,positioning}
\tikzset{
	-Latex,auto,node distance =1 cm and 1 cm,semithick,
	state/.style ={circle, draw, minimum width = 0.7 cm},
	point/.style = {circle, draw, inner sep=0.04cm,fill,node contents={}},
	bidirected/.style={Latex-Latex,dashed},
	el/.style = {inner sep=2pt, align=left, sloped}
}
\usetikzlibrary{positioning}
\usetikzlibrary{fadings}
\usetikzlibrary{intersections}
\usepackage{kantlipsum}
\allowdisplaybreaks

\begin{document}
  \articletype{Research Article{\hfill}Open Access}

  \author*[1]{Michael Lingzhi Li}

\author[2]{Kosuke Imai}

  \affil[1]{Harvard Business School; E-mail: mili@hbs.edu}

  \affil[2]{Harvard University; E-mail: imai@harvard.edu}

  \title{\huge Neyman Meets Causal Machine Learning: Experimental
    Evaluation of Individualized Treatment Rules}

  \runningtitle{Neyman Meets Causal Machine Learning}

  %\subtitle{...}

  \begin{abstract}
    {A century ago, Neyman showed how to evaluate the efficacy of
      treatment using a randomized experiment under a minimal set of
      assumptions.  This classical repeated sampling framework serves
      as a basis of routine experimental analyses conducted by today's
      scientists across disciplines.  In this paper, we demonstrate
      that Neyman's methodology can also be used to experimentally
      evaluate the efficacy of individualized treatment rules (ITRs),
      which are derived by modern causal machine learning algorithms.
      In particular, we show how to account for additional uncertainty
      resulting from a training process based on cross-fitting.  The
      primary advantage of Neyman's approach is that it can be applied
      to any ITR regardless of the properties of machine learning
      algorithms that are used to derive the ITR. We also show,
      somewhat surprisingly, that for certain metrics, it is more
      efficient to conduct this \emph{ex-post} experimental evaluation
      of an ITR than to conduct an \emph{ex-ante} experimental
      evaluation that randomly assigns some units to the ITR.  Our
      analysis demonstrates that Neyman's repeated sampling framework
      is as relevant for causal inference today as it has been since
      its inception.}
\end{abstract}
 % \communicated{...}
 % \dedication{...}

  \journalname{Journal of Causal Inference}
\DOI{DOI}
  \startpage{1}
  \received{..}
  \revised{..}
  \accepted{..}

  \journalyear{2023}
  \journalvolume{1}
%  \journalissue{1}

\maketitle
\section{Introduction}

Neyman's seminal 1923 paper introduced two foundational ideas in
causal inference \cite{neyman1923application}.  First, Neyman
developed a formal notation for potential outcomes and defined the
average treatment effect (ATE) as a causal quantity of interest.
Second, he showed how randomization of treatment assignment alone can
be used to establish the unbiasedness and estimation uncertainty of
the standard difference-in-means estimator.  Since then, combined with
the additional assumption of random sampling of units, Neyman's
repeated sampling framework has served as a basis of routine
experimental analyses conducted by scientists across many disciplines.

Over the past two decades, however, the causal inference literature
has gone beyond the ATE.  Specifically, the realization that the same
treatment can have varying impacts on different individuals led to the
development of statistical methods and machine learning algorithms for
estimating heterogeneous treatment effects
\citep[e.g.,][]{imai:ratk:13,athe:imbe:16,wage:athe:18,hahn:murr:carv:20}.
Furthermore, a number of researchers have developed various methods
for deriving data-driven individualized treatment rules (ITRs)
\citep[e.g.][]{dudi:etal:11, zhan:etal:12a,chak:labe:zhao:14,
  jian:li:16, kall:18, qi2020multi, mo2022efficient, benm:etal:21}.
With an increasing availability of granular data and modern computing
power, these ITRs are becoming popular in business, medicine,
politics, and even public policy.

In this paper, we demonstrate that Neyman's repeated sampling
framework is still relevant for today's causal machine learning
methods. We show how the framework can be used to experimentally
evaluate the efficacy of any ITRs (including those obtained with
machine learning algorithms via cross-fitting) under a minimal set of
assumptions.  While some of our formal results are originally derived
in our previously published work \citep{imai2021experimental} or
follow directly from them, we focus on the intuition behind those
theoretical results to facilitate the future extensions to other
settings.

We also show, using Neyman's framework, that it is not always
statistically more efficient to evaluate an ITR by conducting a new
randomized experiment where the treatment is the administration of the
ITR itself (i.e., {\it ex-ante} evaluation) than simply using the data
from an existing randomized controlled trial (i.e., {\it ex-post}
evaluation).  All together, this paper shows how Neyman's classical
methodological framework can be applied to solve today's causal
inference problems.

\section{Neyman's Repeated Sampling Framework}

We begin by briefly introducing Neyman's inferential approach to
estimating the average treatment effect (ATE).  Suppose that we have a
sample of $n$ units and for each unit we define two potential
outcomes, $Y_i(1)$ and $Y_i(0)$, under the treatment and control
conditions, respectively. Let $T_i$ denote the binary treatment
assignment variable, which is the $i$th element of $n$-dimensional
treatment vector $\bT$.  Then, the observed outcome can be written as
$Y_i=T_iY_i(1)+(1-T_i)Y_i(0)$.  Finally, $\bX_i$ denotes a set of
observed pre-treatment covariates for unit $i$ where $\cX$ represents
the support of covariate distribution.

As pointed out by Rubin in his discussion of Neyman's 1923 paper
\cite{rubi:90}, the above setup implicitly assume no interference
between units --- the outcome of one unit is not influenced by the
treatment of another unit.  We explicitly state this assumption below.
\begin{assumption}[No Interference between Units] \label{asm:SUTVA}
	The potential outcomes for unit $i$ do not depend
	on the treatment status of other units.  That is, for all
	$t_1, t_2,\ldots,t_n \in \{0, 1\}$, we have,
	$$Y_i(T_1 = t_1, T_2 = t_2, \ldots, T_n = t_n) \ = \ Y_i(T_i = t_i).$$
\end{assumption}

Neyman considered the classical randomized experiment where the
treatment assignment is completely randomized with $n_1$ units
assigned to the treatment condition and the remaining $n_0=n-n_1$
units assigned to the control condition.
\begin{assumption}[Complete Randomization of Treatment
  Assignment] \label{asm:comrand} The
  treatment assignment probability is given by,
  $$\Pr(\bT = \bt \mid \{Y_i(1), Y_i(0), \bX_i\}_{i=1}^n ) \ = \ \frac{1}{{n \choose n_1}}$$
  for each $\bt$ where $\sum_{i=1}^n t_i = n_1$.
\end{assumption}

Under these two assumptions alone, Neyman showed the following sample
average treatment effect (SATE) can be estimated without bias,
\begin{equation*}
  \tau_{\text{SATE}} \ = \ \frac{1}{n} \sum_{i=1}^n \{Y_i(1) - Y_i(0)\}.
\end{equation*}
using the difference-in-means estimator $\hat\tau$,
\begin{equation*}
  \E(\hat\tau \mid \{Y_i(1), Y_i(0)\}_{i=1}^n) \ = \
  \tau_{\text{SATE}} \quad \text{ where } \hat\tau = \frac{1}{n_1}\sum_{i=1}^n
T_iY_i - \frac{1}{n_0} \sum_{i=1}^n (1-T_i)Y_i.
\end{equation*}
Neyman also showed that the variance of this estimator is not
identifiable but a conservative variance can be estimated from the
data without bias,
\begin{equation*}
  \V(\hat\tau \mid \{Y_i(1), Y_i(0)\}_{i=1}^n) \ = \
  \frac{1}{n}\left(\frac{n_0}{n_1}S_1^2 + \frac{n_1}{n_0}S_0^2 + 2 S_{01}\right)\  \leq  \ \frac{1}{n}\left(\frac{n_0}{n_1}S_1^2 + \frac{n_1}{n_0}S_0^2 + (S_1^2+ S_0^2)\right) \ = \
  \frac{S_1^2}{n_1} + \frac{S_0^2}{n_0}.
\end{equation*}
where $S_t^2 = \sum_{i=1}^n (Y_i(t) - \overline{Y(t)})^2/(n-1)$,
$S_{01}=\sum_{i=1}^n (Y_i(0) -
\overline{Y(0)})(Y_i(1)-\overline{Y(1)})/(n-1)$ and
$\overline{Y(t)}=\sum_{i=1}^n Y_i(t)/n_t$ for $t=0,1$.

Neyman obtained the above results by averaging over all possible
treatment assignments under complete randomization.  Subsequent work
has extended Neyman's framework to a superpopulation framework by
assuming that the sample of $n$ units are obtained, through random
sampling, from a superpopulation of infinite size $\mathcal{P}$.
\begin{assumption}[Random Sampling of Units] \label{asm:randomsample}
	Each of $n$ units, represented by a three-tuple
	consisting of two potential outcomes and pre-treatment covariates,
	is assumed to be independently sampled from a super-population
	$\mathcal{P}$, i.e.,
	$$(Y_i(1), Y_i(0), \bX_i) \ \iid \ \mathcal{P}.$$
\end{assumption}
This extended framework, which we call Neyman's repeated sampling
framework, is useful because it allows us to estimate the population
average treatment effect from the sample,
\begin{equation*}
  \tau_{\text{PATE}} \ = \ \E(Y_i(1)-Y_i(0)).
\end{equation*}
Subsequent work has shown that the difference-in-means estimator is
unbiased for the PATE and the exact variance can be estimated without
bias \citep{ding2017bridging}.
\begin{eqnarray}
  \E(\hat\tau) & = & \E[\E(\hat\tau \mid \{Y_i(1), Y_i(0)\}_{i=1}^n)]
                     \ = \ \E[\tau_{\text{SATE}}] \ = \ \tau_{\text{PATE}}, \label{eq:diffmean}\\
  \V(\hat\tau) & = & \E[\V(\hat\tau \mid \{Y_i(1), Y_i(0)\}_{i=1}^n)]
                     + \V[\E(\hat\tau \mid \{Y_i(1),
                     Y_i(0)\}_{i=1}^n)] \ = \ \frac{\sigma_1^2}{n_1} + \frac{\sigma_0^2}{n_0},\label{eq:diffvar}
\end{eqnarray} 
where $\sigma_t^2 = \V(Y_i(t))$ for $t=0,1$.

In the remainder of this article, we will show that this Neyman's
repeated sampling framework enables an assumption-free experimental
evaluation of data-driven individualized treatment rules.

\section{Experimental evaluation of individualized treatment rules}
\label{sec:fixed}

In this section, we explain how Neyman's repeated sampling framework
can be applied to experimentally evaluate the empirical performance of
individualized treatment rules (ITRs), which assigns each individual
unit to either the treatment or control condition based on their
observed characteristics.

\subsection{Setup}

Suppose that we use a machine learning (ML) algorithm to create an
individualized treatment rule (ITR), 
\begin{equation*}
	f: \cX \longrightarrow \{0, 1\}.
\end{equation*}
Most commonly, researchers first estimate the conditional average
treatment effect \citep[see e.g.,][]{qian:murp:11,zhan:etal:12a,
	lued:vand:16a,lued:vand:16, zhou:etal:17,kita:tete:18}:
$$\tau(\bx) := \E[Y_i(1)-Y_i(0) \mid \bX_i = \bx].$$
They then
derive an ITR as the treatment rule that assigns the treatment to
everyone who is predicted to have a positive CATE, i.e.,
$f(\bx) = \mathbf{1}\{\hat\tau(\bx) > 0\}$ where $\hat\tau(\bx)$ is an
estimate of the CATE.  Throughout this paper, we assume, without loss
of generality, that a positive effect implies that the treatment is
beneficial.  One may also consider a cost associated with the
administration of treatment or a budget constraint that limits the
proportion of individuals who can receive the treatment.

Our goal is to evaluate the empirical performance of an ITR without
assuming that an ITR is indeed optimal.  In the above example, we do
not assume that $\hat\tau(\bx)$ is an accurate estimate of the CATE.
In fact, we do not make any assumption about how the ITR is
constructed and how accurate it is.  The ITR may be derived from an
application of an ML algorithm or be even based on heuristics.  For
now, we only assume that the ITR to be evaluated is given.  For
example, it may be estimated from an external data set to be used in a downstream decision-making context.  In
Section~\ref{sec:cross-validation}, we discuss how to use the same
experimental data for both learning and evaluating an ITR.

To measure the performance of an ITR, we consider two quantities.  The
first is the Population Average Value (PAV), which is defined as:
\begin{equation}
  \lambda_f := \E[Y_i(f(\bX_i))]. \label{eq:PAV}
\end{equation}
This is the standard metric of ITR's overall performance.  The second
quantity is the Population Average Prescriptive Effect (PAPE)
\cite{radcliffe2007using,imai2021experimental}, which measures the
benefit of individualizing treatment rule and is defined as follows,
\begin{equation}
  \tau_f := \E[Y_i(f(\bX_i))]-p_f\E[Y_i(1)] - (1-p_f)\E[Y_i(0)], \label{eq:PAPE}
\end{equation}
where $p_f = \Pr(f(\bX_i) = 1)$ represents the proportion of
individuals who are treated by the ITR $f$.

The PAPE compares the performance of ITR against the
non-individualized treatment rule that treats the same proportion of
randomly selected individuals.  This contrasts with other quantities
considered in the literature such as the targeting operator
characteristic (TOC) \citep{yadlowsky2021evaluating}, which compares
the performance of ITR against the non-individualized rule that treats
everyone.  Unlike the TOC, the PAPE focuses on the benefit of
determining which individuals should be treated while holding the
proportion of those who receive the treatment constant.

To gain additional intuition about the PAPE, consider the following
alternative but equivalent expression of the same quantity,
\begin{equation*}
  \tau_f \ = \ \Cov(f(\bX_i), Y_i(1)-Y_i(0)).
\end{equation*}
This alternative expression shows that the PAPE measures how well the
ITR agrees with the true individual treatment effect (ITE).
To compare across datasets, we can further normalize the PAPE
as the correlation between the ITR and the true ITE, i.e.,
\begin{equation*}
  \frac{\tau_f}{\sqrt{\V(f(\bX_i))\V(Y_i(1)-Y_i(0))}} = \Corr(f(\bX_i), Y_i(1)-Y_i(0)). 
\end{equation*}
Although this provides a scale invariant quantity to understand the
performance of ITR, it is not identifiable from the data because we
cannot identify the variance of ITE, i.e.,
$\V(Y_i(1)-Y_i(0))$.

The above equality further implies the following inequality by applying Cauchy-Schwarz twice: 
\begin{equation}
  \tau_f\leq \sqrt{\V(f(\bX_i))\V(Y_i(1)-Y_i(0))} \leq \sqrt{2p_f(1-p_f)(\V(Y_i(1)) + \V(Y_i(0)))}.\label{eq:PAPEbound}
\end{equation}  
Therefore PAPE is bounded provided that the second moments of the potential outcomes
exist. Thus, given a fixed variance of the potential outcomes, $\tau_f$ is
most likely largest around $p_f=0.5$. This is because when $p_f$ is
around $0.5$, the ITR has the greatest room to deviate from the
randomized treatment rule.

\subsection{Estimation and Inference}

To estimate the PAV and PAPE under Neyman's repeated sampling
framework, we consider the following ``difference-in-means'' type
estimators:
\begin{align}
	\hat\lambda_f(\bZ_n) \ &= \frac{1}{n_1} \sum_{i=1}^n Y_i T_i f(\bX_i) +
	\frac{1}{n_0} \sum_{i=1}^n Y_i(1-T_i)(1-f(\bX_i)),  \label{eq:PAVest}\\
	\hat\tau_f(\bZ_n) \ &= \ \frac{n}{n-1}\l[\frac{1}{n_1} \sum_{i=1}^n Y_i T_i f(\bX_i) +
	\frac{1}{n_0} \sum_{i=1}^n Y_i(1-T_i)(1-f(\bX_i)) -
	\frac{\hat{p}_f}{n_1} \sum_{i=1}^n Y_iT_i -
	\frac{1-\hat{p}_f}{n_0} \sum_{i=1}^n Y_i(1-T_i)\r],  \label{eq:PAPEest}
\end{align}
where $\hat{p}_f = \sum_{i=1}^n f(\bX_i)/n$ is the estimated
population proportion of individuals who receive the treatment
assignment and $\bZ_n=\{Y_i, T_i, \bX_i\}_{i=1}^n$ represents the
experimental data of sample size $n$.  The factor $n/(n-1)$ in the
PAPE estimator represents the loss in one degree-of-freedom due to
estimating the population-level quantity $p_f$.

The following theorems, reproduced from our previously published work
\cite{imai2021experimental}, shows that, under Neyman's repeated
sampling framework, these two estimators are unbiased and the
finite-sample variances can be derived.
\begin{theorem}[Unbiasedness and Variance of the PAV
	Estimator \citep{imai2021experimental}] \label{thm:PAVest}Under
	Assumptions~\ref{asm:SUTVA},~\ref{asm:comrand}, and~\ref{asm:randomsample},
	the expectation and variance of the PAV estimator defined
	equation~\eqref{eq:PAVest} are given by,
	\begin{eqnarray*}
		\E\{\hat\lambda_f(\bZ_n)\} & = & \lambda_f, \\
		\V\{\hat\lambda_f(\bZ_n)\} & = &  \frac{\E(S_{f1}^2)}{n_1} +
		\frac{\E(S_{f0}^2)}{n_0},
	\end{eqnarray*}
	where
	%    \tau &  = & \E(Y_i(1) - Y_i(0)), \quad
	$S_{ft}^2 = \sum_{i=1}^n (Y_{fi}(t) -
	\overline{Y_f(t)})^2/(n-1)$ with
	$Y_{fi}(t) = \mathbf{1}\{f(\bX_i)=t\}Y_i(t)$, and
	$\overline{Y_f(t)} = \sum_{i=1}^n
	Y_{fi}(t)/n$ for $t=\{0,1\}$.
\end{theorem}
\begin{theorem}[Unbiasedness and Variance of the PAPE Estimator
  \citep{imai2021experimental}] \label{thm:PAPEest}Under
  Assumptions~\ref{asm:SUTVA},~\ref{asm:comrand},
  and~\ref{asm:randomsample}, the expectation and variance of the PAPE
  estimator defined equation~\eqref{eq:PAPEest} are given by,
  \begin{eqnarray*}
    \E\{\hat\tau_f(\bZ_n)\} & = & \tau_f, \\
    \V\{\hat\tau_f(\bZ_n)\} & = &  \frac{n^2}{(n-1)^2}\l[\frac{\E(\widetilde{S}_{f1}^2)}{n_1} +
                                \frac{\E(\widetilde{S}_{f0}^2)}{n_0} + \frac{1}{n^2} \l\{\tau_f^2 -n p_f(1-p_f)
                                \tau^2 + 2(n-1)(2p_f-1) \tau_f\tau \r\} \r],
  \end{eqnarray*}
  where
  % \tau &  = & \E(Y_i(1) - Y_i(0)), \quad
  $\widetilde{S}_{ft}^2 = \sum_{i=1}^n (\widetilde{Y}_{fi}(t) -
  \overline{\widetilde{Y}_f(t)})^2/(n-1)$ with
  $\widetilde{Y}_{fi}(t) = (f(\bX_i) - \hat{p}_f)Y_i(t)$, and
  $\overline{\widetilde{Y}_f(t)} = \sum_{i=1}^n
  \widetilde{Y}_{fi}(t)/n$ for $t=\{0,1\}$.
\end{theorem}

The properties of the PAV estimator shown above follow immediately
from Neyman's classic results by replacing the potential outcome
$Y_i(t)$ with the the potential outcome that incorporates the ITR,
i.e., $\mathbf{1}\{f(\bX_i)=t\}Y_i(t)$.  The form of the estimator
does not mean, however, that we are ignoring observations whose
prescribed treatment status differs from the observed one, i.e.,
$f(\bX_i) \ne T_i$.  These observations are still used when estimating
the variance.  An important implication of this subtle fact is
discussed in Section~\ref{subsec:minvar}.

We can also compare the results of the PAPE estimator with Neyman's
classic results.  We observe that the relevant potential outcome is
given by $\widetilde{Y}_{fi}(t) = (f(\bX_i)-\hat{p}_f)Y_i(t)$, which
directly compares the ITR with the randomized treatment rule.  When
compared to the PAV estimator, the variance of the PAPE estimator has
an additional term, which comes from the fact that
$\widetilde{Y}_{fi}(t)$ is correlated across observations due to the
estimation of the proportion $p_f$.

The correlation can be broadly decomposed into two components: (1) a
negative component caused by the negative correlation within the
mean-adjusted ITR $(f(\bX_i)-\hat{p}_f)$ and $(f(\bX_j)-\hat{p}_f)$,
and (2) the remaining component that arises due to the interaction
between $f(\bX_i)$ and $\tau_i$.  The negative component generally
dominates if the ITR of interest treats roughly $50\%$ of the
population ($p_f \approx 0.5$) and the average treatment effect $\tau$
is not too small. This suggests that the mean adjustment could lead to
a variance reduction when the ITR treats a roughly half of the
population. 

On the other hand, for $p_f \not\approx 0.5$ and $n \gg 1$, this
additional term is positive if and only if the following condition
holds:
\begin{equation*}
  |\tau_f| \gtrsim
  \frac{np_f(1-p_f)}{(n-1)|2p_f-1|}|\tau|.
\end{equation*}
Thus, this additional term is only likely to be positive under a
scenario where $p_f$ is away from 0.5 and the magnitude of the PAPE is
much greater than that of the ATE (i.e., the ITR is performing
well). Equation~\eqref{eq:PAPEbound} suggests that this is unlikely
unless the variance of the individualized treatment effect is large,
implying a large degree of treatment effect heterogeneity.

\subsection{Performance Comparison among Multiple ITRs}
\label{subsec:multiple}

While the PAPE compares an ITR with a random treatment assignment rule
that treats the same proportion of units, researchers are often
interested in comparing the performance of multiple ITRs.  In such
cases, we recommend estimating the difference in PAV between two ITRs
that are subject to the same budget constraint.
\citep{imai2021experimental} provides the details of estimation and
inference regarding this quantity.

The use of PAPE for comparison of two ITRs are likely to be
inappropriate.  To see this, note that the difference in PAPE can be
written as the covariance between the agreement of two ITRs and the
true ITE:
\begin{equation*}
	\tau_f-\tau_g \ = \ \Cov(f(\bX_i)-g(\bX_i), Y_i(1)-Y_i(0)).
\end{equation*}
This expression shows that for ITRs with similar treatment
proportions, the sign of this difference indicates the relative
capability of the ITRs in identifying the optimal individuals to
treat. However, if the two ITRs $f$ and $g$ have significantly
different treatment proportions, then this comparison is difficult.
Figure~\ref{fig:pape_diff} shows an example, in which ITR $f$ has a
higher PAV than ITR $g$, but $f$ has a negative PAPE while $g$ has a
positive PAPE. In this case, $f$ is not an effective ITR as it
performs significantly worse than random treatment, but practitioners
might still choose to not use $g$ as it is able to identify a small
percentage of good patients to target.

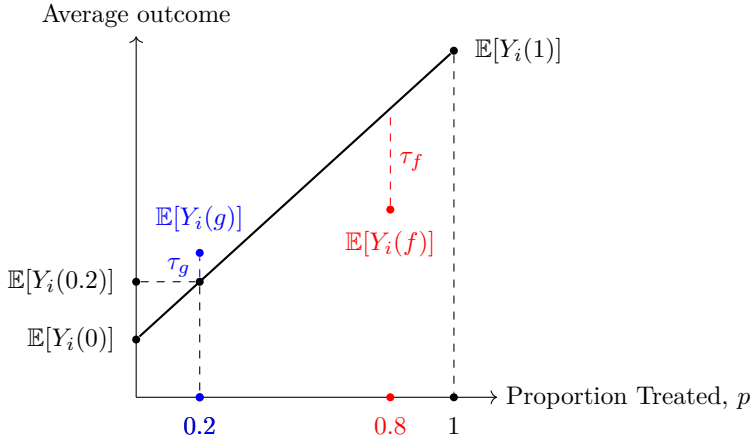
\begin{figure}[h]
	\begin{tikzpicture}[scale=1.9]
	\draw[->] (0,0) -- (2.5,0) coordinate (x axis)  node[right, black]{Proportion Treated, $p$};
	\draw[->] (0,0) -- (0,2.5) coordinate (y axis) node[above, black]{Average outcome};
	\node[circle,fill=black,inner sep=0pt,minimum size=3pt] (a) at (0,0.4) {};
	\node[left=0.1cm of a] {$\E[Y_i(0)]$};	
	\node[circle,fill=black,inner sep=0pt,minimum size=3pt] (b) at (2.2,2.4) {};
	\node[right=0.1cm of b] {$\E[Y_i(1)]$};
	\node[circle,fill=black,inner sep=0pt,minimum size=3pt] (bx) at (2.2, 0) {};	
	\node[below=0.1cm of bx] {$1$};	
	\draw[ dashed,-] (b) -- (bx);	
	\draw[black,thick,-] (a) --  (b);				
	\node[circle,fill=black,inner sep=0pt,minimum size=3pt] (cx) at (0.44, 0) {};
	\node[below=0.1cm of cx] {$0.2$};
	\node[circle,fill=black,inner sep=0pt,minimum size=3pt] (c) at (0.44, 0.8) {};	
	\node[circle,fill=black,inner sep=0pt,minimum size=3pt] (cy) at (0, 0.8) {};
	\node[left=0.1cm of cy] {$\E[Y_i(0.2)]$};
	\draw[dashed,-] (cy) -- (c) -- (cx);			
	\node (rd) at (0.9,1.15) {};
	\node[circle,fill=blue,inner sep=0pt,minimum size=3pt] (g) at (0.44, 1) {};
	\node (gy) at (0, 1) {};
	\node[above=0.1cm of g] {\color{blue} $\E[Y_i(g)]$};	
	\node[circle,fill=blue,inner sep=0pt,minimum size=3pt] (gx) at (0.44, 0) {};
	\node[below=0.1cm of gx] {\color{blue} $0.2$};
	\draw[blue, dashed,-] (c) -- (g) node [midway, left] {$\tau_g$};			
	\node[circle,fill=red,inner sep=0pt,minimum size=3pt] (f) at (1.76, 1.3) {};
	\node (fy) at (0, 1.4) {};
	\node[below=0.1cm of f] {\color{red} $\E[Y_i(f)]$};
	\node[circle,fill=red,inner sep=0pt,minimum size=3pt] (fx) at (1.76, 0) {};
	\node[below=0.1cm of fx] {\color{red} $0.8$};
	\node (fl) at (1.76, 2.0) {};
	\draw[red, dashed,-]  (fl) -- (f) node [midway, right] {$\tau_f$};	
\end{tikzpicture}
\caption{Illustration of PAPE for two different ITRs $f$ and $g$. Here the $x$ axis is the proportion of individuals treated, and $y$ axis is PAV. The PAV of $f$ is higher than PAV of $g$, but ITR $g$ has a positive PAPE $\tau_g$ and the ITR $f$ has a negative PAPE $\tau_f$. }
\label{fig:pape_diff}
\end{figure}

\subsection{Lack of invariance}
\label{subsec:minvar}

Unlike Neyman's ATE estimator, the PAV and PAPE estimators are not
invariant to a constant shift of the outcome variable.  One might
expect that adding a constant $\delta$ to $Y$ would shift the PAV
estimator by $\delta$ and not affect the PAPE estimator at all.
Unfortunately, this is not the case.  For both of these estimators, a
constant shift of the outcome will result in an \emph{additional}
change of the equal magnitude.  Let $\hat\lambda^\delta_f(\bZ)$ and
$\hat\tau^\delta_f(\bZ)$ be the new PAV and PAPE estimators under a
constant $\delta$ shift, i.e.,
$\hat\lambda^\delta_f(\bZ) = \hat\lambda_f(\bZ)+\delta$ and
$\hat\tau^\delta_f(\bZ)= \hat\tau_f(\bZ)+\delta$. Then, we have,
\begin{equation*}
  \hat\lambda^\delta_f(\bZ) -\hat\lambda_f(\bZ)  - \delta \ = 	\frac{n-1}{n}\left(\hat\tau^\delta_f(\bZ) -\hat\tau_f(\bZ)\right)
  \ = \ \delta\left(\frac{1}{n_1}\sum_{i=1}^n T_if(\bX_i) + \frac{1}{n_0}\sum_{i=1}^n (1-T_i)(1-f(\bX_i)) -1\right).
\end{equation*}

Since this term equals zero in expectation, the two estimators remain
unbiased.  However, this shift affects their variances because the ITR
$f(\bX_i)$ is only balanced on average due to the randomized treatment
assignment. Intuitively, this is because the relevant potential
outcomes in the case of both PAV and PAPE are not invariant to a
constant shift of the outcome variable. For PAV, the relevant
potential outcome is $Y_{fi}(t)=\mathbf{1}\{f(\bX_i)=t\}Y_i(t)$, so
under a constant shift of $\delta$, only samples with $f(\bX_i)=t$
would be shifted by $\delta$ while the remaining samples still has a
value of zero. The intuition is similar for PAPE. Therefore, we expect
that balancing the observed outcomes $Y_i$ close to zero could result
in increased efficiency.
%One method to correct for this would be to utilize $k$ additional samples of $Y$, $Y_{n+1},\cdots, Y_{n+k}$ and add the correction term:
%\begin{equation}
%	-\frac{1}{k} \sum_{i=1}^k Y_{n+i}\left(\frac{1}{n_1}\sum_{i=1}^n T_if(\bX_i) + \frac{1}{n_0}\sum_{i=1}^n (1-T_i)(1-f(\bX_i)) -1\right)
%\end{equation}
%However, such a strategy is generally not efficient in utilizing the
%available samples, and also increases the estimator variance.

We now derive the constant shift $\delta$ that minimizes the resulting
variances.  We can show that a constant shift of $\delta$ to potential
outcomes creates the following additional variance terms for both
estimators:
\begin{proposition}[Minimum Variance
  Estimators] \label{prop:invariance} The variances of the constant
  shift estimators are given by:
  \begin{align*}
    \V(\hat\lambda^\delta_f(\bZ))  & \ = \ \V(\hat\lambda_f(\bZ)) +
                                     \delta
                                     p_f(1-p_f)\left(\frac{2\kappa_{11}}{n_1}+
                                     \frac{2\kappa_{00}}{n_0}+ \delta \cdot \frac{n}{n_1n_0}\right),\\
    \V(\hat\tau^\delta_f(\bZ)) & \ = \ \V(\hat\tau_f(\bZ)) +
                                 \frac{n^2}{(n-1)^2} \delta
                                 p_f(1-p_f)\left(\frac{2\kappa_{11}}{n_1}+
                                 \frac{2\kappa_{00}}{n_0}+ \delta \cdot \frac{n}{n_1n_0}\right) + O\left(\frac{\delta}{n^2}\right),
  \end{align*}
  where $\kappa_{st}= \E[Y_i(s) \mid f(\bX_i)=t]$. Minimizing these
  variances over $\delta$ results in the following optimal value of
  $\delta^\ast_\lambda$ and $\delta^\ast_\tau$ for the PAV and PAPE,
  respectively:
\begin{equation*}
	\delta^\ast_\lambda \ = \
        -\left(\frac{n_0}{n}\kappa_{11}+\frac{n_1}{n}\kappa_{00}\right);
        \quad \delta^\ast_\tau  = \delta^\ast_\lambda + O\left(\frac{1}{n}\right).\label{eq:optimal-delta}
\end{equation*}
\end{proposition}
Proof is in Appendix~\ref{app:invariance}.  The proposition implies
that if we wish to minimize variance across a range of $f(\bX_i)$,
then when $n_1=n_0=n/2$, the optimal value of $\delta$ approximately
balances the two potential outcomes around zero after a constant
shift, i.e.,
$$\frac{1}{n}\sum_{i=1}^n \{(Y_i(1) + \delta^\ast_\lambda) + (Y_i(0) +
\delta^\ast_\lambda)\} \approx 0.$$

\section{Ex-ante vs. ex-post experimental evaluations}
\label{subsec:expostexantedisc}

So far, we have considered an \textit{ex-post} evaluation, in which we
first conduct a completely randomized experiment and then evaluate
ITRs using the data from the experiment.  Alternatively, researchers
may consider an {\it ex-ante} experimental evaluation, in which we
randomly assign units to an ITR.  That is, the ITR itself is the
``treatment'' of this experiment.  Ex-ante experimental designs are
commonly used in practice \citep[see][for
example]{kumar2020orderrex,forman2019randomized}.

We apply the Neyman's repeated sampling framework to compare the
statistical efficiency of \textit{ex-ante} and {\it ex-post}
experimental evaluations.  We show below that, perhaps surprisingly,
in some cases \textit{ex-post} evaluation is more efficient than
\textit{ex-ante} evaluation.  Our result suggests that given a
potential ethical concern of {\it ex-ante} experimental evaluation,
researchers may prefer {\it ex-post} evaluation.  Another reason to
prefer {\it ex-post} evaluation is that this design allows one to
evaluate any number of ITRs while the {\it ex-ante} evaluation is tied
to a particular ITR. In this section, our analysis focuses on the
PAPE.  Since the PAV does not compare between two different treatment
regimes, it does not make sense to design a randomized trial around
it.

\subsection{Setup}
\label{subsec:comparison_setup}

For the {\it ex-ante} evaluation of the PAPE, we assume a simple
random of $n$ units from the same target population, $\mathcal{P}$.
Consider a completely randomized experiment, in which a total of $n_f$
units are randomly assigned to an ITR $f$ while the remaining units
$n_r = n-n_f$ are assigned to the random treatment rule with the
probability of treatment assignment equal to $n_{r1}/n_r$.  Let $F_i$
be an indicator variable, which is equal to 1 if unit $i$ is assigned
to the ITR $f$ and is equal to 0 otherwise.  Under the random
treatment rule, the number of units that are randomly assigned to the
treatment condition is $n_{r1}$ while $n_{r0} = n_r - n_{r1}$ units
are assigned to the control condition.  As before, we use $T_i$ to
represent the treatment indicator.  We formally state these
assumptions.
\begin{assumption}[Complete Randomization in the Ex-ante Evaluation of PAPE] \label{asm:exante_random}
The probability of being
	assigned to the individualized treatment rule rather than the random
	treatment rule is given by,
	$$\Pr(\bF = \bm{f}\mid \{Y_i(1), Y_i(0), \bX_i\}_{i=1}^n ) \ = \ \frac{1}{ { n \choose n_f}},$$
	for each $\bm{f}$ where $\sum_{i=1}^n f_i = n_f$.  Among those
        who are assigned to the random treatment rule, i.e.,
        $F_i = 0$, the probability of treatment assignment is given
        by,
	$$\Pr(\bT = \bt \mid \{Y_i(1), Y_i(0), \bX_i\}_{i=1}^n) \ = \
	\frac{1}{ { n_r\choose n_{r1}}},$$
	for each $\bt$ where $\sum_{i=1}^n (1-F_i)t_i = n_{r1}$.
\end{assumption}

Using this experimental data, we wish to estimate the PAPE defined in
Equation~\eqref{eq:PAPE}.  For simplicity, we have the number of
treated units under the random treatment rule to equal that under the
ITR condition, i.e., $\hat{p}_f = n_{r1} / n_r$ where
$\hat{p}_f=\sum_{i=1}^n f(\bX_i)/n$, so that the evaluation estimator
would not need to be further adjusted. Fortunately, in practice, this
can be easily accomplished so long as the covariates are available
prior to the randomization of treatment assignment among the group
assigned to the random treatment rule.  Lastly, the so-called Neyman
allocation implies that if the variances of $Y_i(1)$ and $Y_i(0)$
differ significantly from one another, one can gain additional
statistical efficiency by allocating more units to the treatment
condition whose potential outcome has a greater variance.  This
optimal design, however, does require the availability of external
data to estimate these variances.  For the sake of simplicity, we do
not consider such optimal designs here.

We consider the following estimator of the PAPE for the
\textit{ex-ante} experimental evaluation that accounts for a potential
difference in the proportion of treated units between the ITR and the
random treatment rule by appropriately weighting the latter,
\begin{align}
	\hat\tau_f^\ast(\bZ_n) \ &= \ \frac{n}{n-1} \l(\frac{1}{n_f}  \sum_{i=1}^n
	F_i Y_i - \frac{\hat{p}_f}{n_{r1}} \sum_{i=1}^n (1-F_i)T_i Y_i -
	\frac{1-\hat{p}_f}{n_{r0}} \sum_{i=1}^n (1-F_i)(1-T_i) Y_i  \r). \label{eq:PAPEaest}
\end{align}

The \textit{ex-ante} evaluation differs from the \textit{ex-post}
evaluation in two ways. First, the \textit{ex-ante} estimator requires
two separate random assignments ($T_i$ and $F_i$) while the
\textit{ex-post} estimator only involves one.  Intuitively, an
additional layer of randomization increases variance. Second, the
\textit{ex-ante} evaluation requires a separate group that follows an
ITR, whereas all individuals under the \textit{ex-post} evaluation are
simply randomly assigned either to the treatment or control group.  As
a result, under the \textit{ex-post} evaluation, we utilize the
samples identically, which could further reduce the variance.
Together, we expect the \textit{ex-ante} evaluation to be less
efficient than the \textit{ex-post} evaluation as the full sample is
not utilized for every part of the estimation.  Below, we use Neyman's
repeated sampling framework to confirm this intuition under a set of
simplifying assumptions.

\subsection{Comparison of the two experimental designs}
\label{subsec:comparison}

Before comparing two modes of evaluation, we derive the bias and
variance of the \textit{ex-ante} evaluation estimators under the
Neyman's repeated sampling framework.  In the current case, the
uncertainty comes from three types of randomness: (1) the random
assignment to the individualized or random treatment rule, (2) the
randomized treatment assignment under the random assignment rule, and
(3) the simple random sampling of units from the target population.
The next theorem shows that this estimator is unbiased and the
variance is identifiable.  Proof is given in
Appendix~\ref{app:PAPEaest}.
\begin{theorem}[Unbiasedness and Variance of the \textit{Ex-ante}  PAPE
	Estimator] \label{thm:PAPEaest}  Under
	Assumptions~\ref{asm:SUTVA},~\ref{asm:randomsample},~and~\ref{asm:exante_random},
	the expectation and variance of the {\it ex-ante} PAPE estimator defined in
	Equation~\eqref{eq:PAPEaest} are given by,
	\begin{eqnarray*}
		\E(\hat\tau_f^\ast(\bZ_n)) & = & \tau_f, \\
		\V(\hat\tau_f^\ast(\bZ_n)) & = &
		\frac{n^2}{(n-1)^2}\l[\E\l\{\frac{S_f^2}{n_f}+\frac{\hat{p}_f^2S_{1}^2}{n_{r1}}+\frac{(1-\hat{p}_f)^2S_{0}^2}{n_{r0}}\r\}\r.\\
		& & \hspace{1in} \l.
		+\frac{1}{n^2} \l\{\tau_f^2 -n p_f(1-p_f)
		\tau^2 + 2(n-1)(2p_f-1) \tau_f\tau \r\}\r],
	\end{eqnarray*}
	where
	$S_f^2= \sum_{i=1}^n (Y_i(f(\bX_i))- \overline{Y(f(\bX))})^2/(n-1)$,
	and $S_{t}^2 = \sum_{i=1}^n (Y_i(t) -\overline{Y(t)})^2/(n-1)$ with
	$\overline{Y(f(\bX))}=\sum_{i=1}^n Y_i(f(\bX_i))/n$ and
	$\overline{Y(t)}=\sum_{i=1}^n Y_i(t)/n$ for $t=0,1$.
\end{theorem}

Given these results, we examine the relative statistical efficiency of
the \textit{ex-post} and \textit{ex-ante} experimental evaluations.
To facilitate the comparison, we assume $n_1 = n_0 = n_f = n_r = n/2$.
In words, the \textit{ex-post} evaluation sets the treatment
assignment probability to $1/2$, and the \textit{ex-ante} evaluation
also sets the probability of being assigned to the ITR to $1/2$.  In
the same fashion, we also assume $n_{r1} = n_{r0} = n/4$, implying
that the \textit{ex-ante} evaluation sets the treatment assignment
probability under the random treatment rule to $1/2$ as well. Although
our result below may not be applicable beyond this simplified setting,
we believe that this equal allocation setting is a common choice in
practice and therefore is worthy of investigation.

Under this simplified setting, the difference in the variance of the
PAPE estimator between the \textit{ex-ante} and \textit{ex-post}
evaluations is given by,
\begin{eqnarray}
	& & \V(\hat{\tau}_f^\ast(\bZ_n))-\V(\hat{\tau}_f(\bZ_n)) \nonumber\\
	& = & \frac{2n}{(n-1)^2}\l[\E\l\{p_f^2S_1^2+(1-p_f)^2S_0^2\r\}
	+2\Cov(f(\bX_i)Y_i(1),(1-f(\bX_i))Y_i(0)) \r. \nonumber\\
	& & + \l. 2 p_f\Cov(f(\bX_i)Y_i(1),Y_i(1)) +
	2(1-p_f)\Cov((1-f(\bX_i))Y_i(0),Y_i(0))\r] \nonumber\\
	& = & \frac{2n}{(n-1)^2}\l[p_f^2\V(Y_i(1))+(1-p_f)^2\V(Y_i(0))
	-2p_f(1-p_f)\E(Y_i(0) \mid f(\bX_i)=0)\E(Y_i(1)\mid f(\bX_i)=1) \r. \nonumber\\
	& & \hspace{.5in} +  2 p_f^2 \l\{\E(Y_i^2(1) \mid f(\bX_i) = 1) -
	\E(Y_i(1))\E(Y_i(1) \mid f(\bX_i) =1)\r\} \nonumber \\
	& & \l. \hspace{.5in} +  2(1-p_f)^2 \l\{\E(Y_i^2(0) \mid f(\bX_i) = 0) -
	\E(Y_i(0))\E(Y_i(0) \mid f(\bX_i) =0)\r\} \r]. \label{eq:var_diff_pape}
\end{eqnarray}
The details of the derivation are given in
Appendix~\ref{app:var_diff_pape}. Suppose now that the ITR correctly
assigns individuals on average, i.e.,
$\E(Y_i(t) \mid f(\bX_i) = t) \ge \E(Y_i(t) \mid f(\bX_i) = 1-t)$ for
$t=0,1$.  Under this assumption, the last two terms in the square
bracket are positive, i.e.,
\begin{equation*}
	\E(Y_i^2(t) \mid f(\bX_i) = t) - \E(Y_i(t))\E(Y_i(t) \mid f(\bX_i) =
	t) \ \ge \ \V(Y_i(t) \mid f(\bX_i) = t),
\end{equation*}
for $t=0,1$.  Hence, the only term that is possibly negative in
Equation~\eqref{eq:var_diff_pape} is the third term in the square
bracket. For simplicity, further assume that we shift the outcomes to minimize
variance of the \textit{ex-post} estimator and achieved
$\E(Y_i(1) + Y_i(0) \mid f(\bX_i)=1)=\E(Y_i(1) + Y_i(0)\mid
f(\bX_i)=0)=0$ (see Equation~\eqref{eq:optimal-delta}).  This
guarantees that the optimal choice of $\delta$ is zero and hence no
adjustment in variance is necessary.  Under this assumption, we can
bound Equation~\eqref{eq:var_diff_pape} from below as follows (see
Appendix~\ref{app:comp_simple} for details),
\begin{eqnarray*}
	& & \V(\hat{\tau}_f^\ast(\bZ_n))-\V(\hat{\tau}_f(\bZ_n))  \\
	& = & \frac{2n}{(n-1)^2}\l[p_f^2\V(Y_i(1))+(1-p_f)^2\V(Y_i(0))
	+2p_f^2\V(Y_i(1) \mid f(\bX_i)=1) + 2(1-p_f)^2\V(Y_i(0) \mid f(\bX_i)=0) \r.\\
	& & \l. +2p_f(1-p_f)\l[(1-p_f)\{\E(Y_i(0)\mid f(\bX_i)=0)\}^2+p_f\{\E(Y_i(1)\mid f(\bX_i)=1)\}^2\r] \r] \\
	& \geq & 0.
\end{eqnarray*}
The result implies that under a set of simplifying assumptions the
\textit{ex-post} evaluation is more efficient than the
\textit{ex-ante} evaluation.  We note, however, that this conclusion
may not hold if the \textit{ex-ante} and \textit{ex-post} setups have
sample allocation different from the setting considered here.

\section{Incorporating the uncertainty of machine learning training}
\label{sec:cross-validation}

In the above sections, we have assume that the ITR to be evaluated is
given.  For example, an ITR may be derived using an external data set.
But, in many cases, researchers may wish to use the same experimental
data set to both derive an ITR and evaluate it.  One possibility is to
randomly split a data set into the training and evaluation data sets,
and then use the former to learn an ITR and the latter for its
evaluation.  Unfortunately, this {\it sample splitting} approach does
not utilize the data most efficiently.

An alternative and more efficient approach is {\it cross-fitting}.
The idea is to randomly split the data into $K$ folds of equal size
and then use each fold as the evaluation data while using the
remaining $K-1$ folds as the training data to learn an ITR.  By
repeating this process across $K$ folds and averaging the evaluation
results, we are able to use the entire data set for both training and
evaluation.

While the dominant ``double machine learning'' (DML) approach uses the
same cross-fitting procedure \citep{chernozhukov2018double}, we show
here that Neyman's repeated sampling framework can also incorporate
this cross-fitting approach.  Unlike the DML, Neyman's framework
enables us to derive the finite-sample properties of ITR evaluation
solely based on the random splitting of the data as well as
randomization of treatment assignment and random sampling of units.

\subsection{Setup}

Consider a generic ML algorithm, which we define as
a deterministic function mapping the space of training data of finite
size, denoted by $\cZ$, to the space of all possible scoring rules
$\cS$,
\begin{equation}
  F: \cZ \to \cS.
\end{equation}
Typically, the scoring rule of interest is the estimated CATE such
that the largest value indicates the highest treatment prioritization.
Alternatively, the scoring rule may be based on the estimated baseline
risk, i.e., $\E(Y_i(0) \mid \bX_i = \bx)$.  We do not, however, assume
that the ML algorithm used to generate the scoring rule accurately
estimates either the CATE or baseline risk.  Indeed, we essentially
impose no assumption on how the scoring rule is created.  Once the
scoring rule is estimated by an ML algorithm, the ITR is given by,
\begin{equation}
  \hat{f}_{\bZ_n}(\bx):=\mathbf{1}\{F(\bZ_n)(\bx)>0\}, \label{eq:f.hat}
\end{equation}
where the notation makes it explicit that the ITR depends on the
specific training data $\bZ_n \in \cZ$ of sample size $n$.

Next, consider the following standard cross-fitting procedure.  First,
we randomly split the experimental data of size $n$ into $K$
subsamples of equal size $m = n/K$ where, for notational simplicity,
we assume $n$ is a multiple of $K$.  Then, for each $k=1,2,\ldots,K$,
we use the $k$th subsample as an evaluation dataset
$\bZ_m^{(k)}=\{\bX_i^{(k)}, T_i^{(k)}, Y_i^{(k)}\}_{i=1}^{m}$ while
the remaining $(K-1)$ subsamples are used as the training dataset
$\bZ_{n-m}^{(-k)}=\{\bX_i^{(-k)}, T_i^{(-k)},
Y_i^{(-k)}\}_{i=1}^{n-m}$.  Without loss of generality, we assume that
the number of treated (control) units is identical across $K$ folds
and denote it using $m_1$ ($m_0=m-m_1$).

Then, for each fold $k$, we estimate an ITR by applying the ML
algorithm $F$ to the training data $\bZ_{n-m}^{(-k)}$, which we denote
by $\hat{f}^{(-k)}=\hat{f}_{\bZ_{n-m}^{(-k)}}$.  We then evaluate the
performance of the ML algorithm $F$ by computing an evaluation metric
of interest based on the test data $\bZ_m^{(k)}$.  Repeating this
process $K$ times for each $k$ and averaging the results gives a
cross-fitting estimator of the evaluation metric.  Here, we focus on
the cross-fitting PAV estimator,
%\begin{algorithm}[ht]
%	\begin{algorithmic}[1]
%		\Procedure{Fixed ITR}{$\cZ=\{\bm{X},T,Y\}$,  $ \hat{\tau}_f$, $f$}
%		\State $ \hat{\tau} =  \hat{\tau}_{f}(\cZ)$
%		\State \textbf{return} $\hat{\tau}$
%		\EndProcedure
%	\end{algorithmic}
%	\caption{Fixed ITR procedure for general evaluation metric $ \hat{\tau}_f$ and data $\cZ=\{\bm{X},\bm{T},\bm{Y}\}$. }
%	\label{alg:fixed_itr}
%\end{algorithm}
\begin{equation}
  \hat{\lambda}_{K}^F(\bZ_n) \ = \ \frac{1}{K}\sum_{k=1}^K \hat{\lambda}_{\hat{f}^{(-k)}}(\bZ_m^{(k)}), \label{eq:PAVcvest}
\end{equation}
where $\hat\lambda_f(\cdot)$ is defined in Equation~\eqref{eq:PAVest}.
We now discuss the estimand, for which this cross-fitting estimator is
unbiased.

\subsection{Evaluation metrics under cross-fitting}

To extend Neyman's repeated sampling framework to cross-fitting with
$K \ge 2$ folds, we begin by noting that the ITR in this setting
varies as a function of training data.  Thus, we consider the
performance measure that averages over the random sampling of training
data as well as the randomization of treatment assignment and random
sampling of units.  In other words, we evaluate the average
performance of ITR that is generated by the application of ML
algorithm $F$ across $K$ different (but overlapping) training data
sets.  This contrasts with the performance evaluation metric of a
fixed ITR discussed in earlier sections.

For the PAV under cross-fitting, we consider an average ITR over
across training data of size $n-m$,
\begin{equation*}
  \bar{f}^{F}_{n-m}(\bX_i) \ = \ \E_{\bZ_{n-m}}\{\hat{f}_{\bZ_{n-m}}(\bX_i) \mid \bX_i \} \ = \ \mathbb{P}_{\bZ_{n-m}} \{\hat{f}_{\bZ_{n-m}}(\bX_i) = 1\mid \bX_i \},
\end{equation*}
which represents the proportion of times the estimated ITR would
assign the treatment to a unit with a specific value of covariates.
The notation makes explicit the dependence on the size of training
data $n-m$ as well as ML algorithm $F$.
      
Under Neyman's repeated sampling framework, one can view each
estimated ITR as another random sampling from a population of ITRs
based on ML algorithm $F$ with training data set of size $n-m$.  Thus,
the PAV under cross-fitting can be defined as,
\begin{equation*}
  \lambda^{F}_{n-m} \ := \ \E\{\bar{f}_{n-m}^{F}(\bX_i) Y_i(1) + (1-\bar{f}_{n-m}^{F}(\bX_i))Y_i(0) \}. \label{eq:PAVcv}
\end{equation*}
For the PAPE, we consider the cross-fitting version of the proportion
treated by ITR $p_f$ as follows:
 \begin{equation*}
 	p^F_{n-m} \ := \ \mathbb{P}_{\bZ_{n-m}} \{\hat{f}_{\bZ_{n-m}}(\bX_i) = 1\}.
 \end{equation*}
Then, the  PAPE under cross-fitting can be defined as: 
\begin{equation}
  \tau^{F}_{n-m} \ := \ \E\{\bar{f}_{n-m}^{F}(\bX_i) Y_i(1) + (1-\bar{f}_{n-m}^{F}(\bX_i))Y_i(0) - p_{n-m}^F Y_i(1) - (1-p_{n-m}^F) Y_i(0)\}. \label{eq:PAPEcv}
\end{equation}
As shown before, the PAPE is equal to the covariance between the
average proportion treated and the individual treatment effect, 
\begin{equation*}
	\tau^{F}_{n-m} \ = \ \Cov(\bar{f}^{F}_{n-m}(\bX_i), Y_i(1)-Y_i(0)).
\end{equation*}

\subsection{Finite sample properties}

We now apply Neyman's repeated sampling framework to the cross-fitting
PAV estimator given in Equation~\eqref{eq:PAVcvest}.  It is easy to
show that $\hat\lambda_K^F$ is an unbiased estimator of
$\lambda^F_{n-m}$.  To derive the variance, we first note that the
evaluation metric is correlated across $K$ folds because cross-fitting
utilizes each subsample for both training and testing,
\begin{equation*}
  \V(\hat{\lambda}_K^F(\bZ_n)) \ = \
  \frac{\V(\hat{\lambda}_{\hat{f}^{(-k)}}(\bZ_m^{(k)}))}{K} + \frac{K-1}{K}
  \Cov( \hat{\lambda}_{\hat{f}^{(-k)}}(\bZ_m^{(k)}),  \hat{\lambda}_{\hat{f}^{(-\ell)}}(\bZ_m^{(\ell)})),
\end{equation*}
where $k \ne \ell$.  We then use a useful lemma about cross-fitting
due to \cite{nadeau2003inference}, and rewrite the covariance term as
follows, 
\begin{equation*}
\Cov( \hat{\lambda}_{\hat{f}^{(-k)}}(\bZ_m^{(k)}),
\hat{\lambda}_{\hat{f}^{(-\ell)}}(\bZ_m^{(\ell)})) \ = \
\V(\hat{\lambda}_{\hat{f}^{(-k)}}(\bZ_m^{(k)})) -  \E(S_F^2), 
\end{equation*}
where $S^2_F$ is the sample variance of
$\hat{\lambda}_{\hat{f}^{(-k)}}(\bZ_m^{(k)})$ across $K$ folds.
Putting them together, we have,
\begin{equation}
  \V(\hat{\lambda}_K^F(\bZ_n)) \ = \
  \V(\hat{\lambda}_{\hat{f}^{(-k)}}(\bZ_m^{(k)})) -\frac{K-1}{K}
  \E(S_F^2). \label{eq:lambdaFvar}
\end{equation}  

We can further analyze the first term of
Equation~\eqref{eq:lambdaFvar} by following the analytical strategy
used in Theorem~\ref{thm:PAVest}.  The only difference is that the
estimated ITR is correlated across observations due to training
process,
\begin{equation}
  \V(\hat{\lambda}_{\hat{f}^{(-k)}}(\bZ_m^{(k)})) \ = \
  \frac{\E(S_{\hat{f}1}^2)}{m_1} + \frac{\E(S_{\hat{f}0}^2)}{m_0} +
  \Cov(Y_{\hat{f}i}(1)-Y_{\hat{f}i}(0), Y_{\hat{f}j}(1) - Y_{\hat{f}j}(0)),    
\end{equation}
where $i \ne j$,
$Y_{\hat{f}i}(t) = \mathbf{1}\{\hat{f}^{(-k)}(\bX_i)=t\}Y_i(t)$, and
$S^2_{\hat{f}t}$ is the sample variance of $Y_{\hat{f}i}(t)$.  Further
simplifying the covariance term yields the following theorem whose proof
is given in Appendix~A.5.1. of our previously published work
\citep{imai2021experimental}.
\begin{theorem} {\sc (Unbiasedness and Exact Variance of the
    Cross-Fitting PAV Estimator
    \cite{imai2021experimental})} \label{thm:PAVcvest} Under
  Assumptions~\ref{asm:SUTVA}--\ref{asm:randomsample} the expectation
  and variance of the cross-fitting PAV estimator defined in
  equation~\eqref{eq:PAVcvest} are given by,
	\begin{eqnarray*}
		\E(\hat{\lambda}_K^F(\bZ_n))  & = & \lambda_{n-m}^F, \\
		\V(\hat{\lambda}_K^F(\bZ_n))
		& = & \frac{\E(S_{\hat{f}1}^2)}{m_1} + \frac{\E(S_{\hat{f}0}^2)}{m_0} +
\E\Big\{\Cov(\hat{f}^{(-k)}(\bX_i), \hat{f}^{(-k)}(\bX_j)
		\mid \bX_i, \bX_j) \tau_i\tau_j \Big\}   -  \frac{K-1}{K}\E(S_{F}^2),
	\end{eqnarray*}
	for $i\neq j$, where $\tau_i = Y_i(1)-Y_i(0)$,
        $S_{\hat{f}t}^2 = \sum_{i=1}^{m} (Y_{\hat{f}i}(t) -
        \overline{Y_{\hat{f}}(t)})^2/(m-1)$,
        $S_{F}^2 = \sum_{k=1}^K
        (\hat\lambda_{\hat{f}^{(-k)}}(\bZ_{m}^{(k)}) -
        \overline{\hat\lambda_{\hat{f}^{(-k)}}(\bZ^{(k)}_m)}^2/(K-1)$
        with
        $Y_{\hat{f}i}(t) =
        \mathbf{1}\{\hat{f}^{(-k)}(\bX_i)=1\}Y_i(t)$,
        $\overline{Y_{\hat{f}}(t)} = \sum_{i=1}^{m}
        Y_{\hat{f}i}(t)/m$, and
        $\overline{\hat\lambda_{\hat{f}^{(-k)}}(\bZ_{m}^{(k)})} =
        \sum_{k=1}^K \hat\lambda_{\hat{f}^{(-k)}}(\bZ_{m}^{(k)})/K$,
        for $t=\{0,1\}$.
\end{theorem}

In particular, we note that when compared with the fixed ITR setting,
there are two additional terms. One of these terms is proportional to
$\Cov(\hat{f}^{(-k)}(\bX_i), \hat{f}^{(-k)}(\bX_j) \mid \bX_i,
\bX_j)$, which represents the covariance between the evaluation
samples due to the training process, and the product of individual
treatment effects $\tau_i\tau_j$.  This term is often positive because
if the ITR is estimated well, it is more likely to make the same
treatment assignment to units when their individual treatment effects
are similar.  In numerical experiments, we typically find that this
term is usually relatively small. The other term
$- \frac{K-1}{K}\E(S_{F}^2)$ is always negative and quantifies the
efficiency gain resulting from utilizing the cross-validation
procedure. Furthermore, from Lemma~1 in \cite{nadeau2003inference}, we
can show that
\begin{equation*}
	\V(\hat{\lambda}_{\hat{f}^{(-k)}}(\bZ_m^{(k)})) \ = \
	\frac{\E(S_{\hat{f}1}^2)}{m_1} +
	\frac{\E(S_{\hat{f}0}^2)}{m_0}
	+\E \l\{\Cov(\hat{f}^{(-k)}(\bX_i),
	\hat{f}^{(-k)}(\bX_j) \mid
	\bX_i, \bX_j)\tau_i\tau_j\r\} \ \geq \ \E(S_{F}^2).
\end{equation*}
Therefore, maximally the efficiency gain resulting from the
cross-fitting procedure reduces the variance to $\E(S^2_F)$ when the
estimated PAV from each of $K$ folds is completely independent. 

\section{A Numerical Study}

In this section, we empirically validate our theoretical results
through a numerical study. In particular, we focus on demonstrating
the results related to the lack of invariance
(Proposition~\ref{prop:invariance}) and the efficiency comparison
between the ex-ante and ex-post estimators
(Theorem~\ref{thm:PAPEaest}).  Strong finite-sample performance of the
proposed estimators have been extensively demonstrated in our
previously published study \citep{imai2021experimental}.

In all our simulations, we utilize the 28th data generating process
(DGP) from the 2016 Atlantic Causal Inference Conference (ACIC)
Competition, of which the details are given in \cite{dori:etal:19}.
For the population distribution of pre-treatment covariates, we use
the empirical distribution of covariates from this sample of $n=4802$
observations with 58 covariates $\bX$ including 3 categorical, 5
binary, 27 count data, and 13 continuous variables. That is, we obtain
each simulation sample via bootstrap. We further assume that the
treatment assignment is completely randomized, and the treatment and
control groups are of equal size, i.e., $n_1=n_0=n/2$.  Finally, the
formula for the outcome model is reproduced in Appendix~\ref{app:sim}.

\begin{figure}[H]
	 \begin{subfigure}{.5\textwidth}
	\centering
	\includegraphics[width=\textwidth]{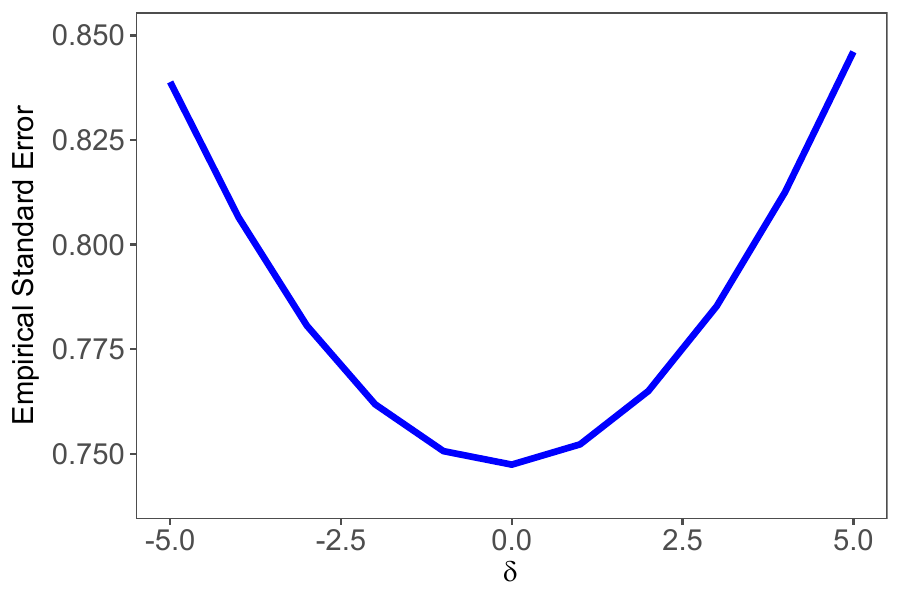}
        \caption{The empirical standard error of PAV estimator as a function of
          constant shift in potential outcomes. $\delta=0$ minimizes
          the standard error of PAV. }
	   	\label{subfig:var_shift}
	 \end{subfigure}%
	 \begin{subfigure}{.5\textwidth}
	   \centering
	   \includegraphics[width=\textwidth]{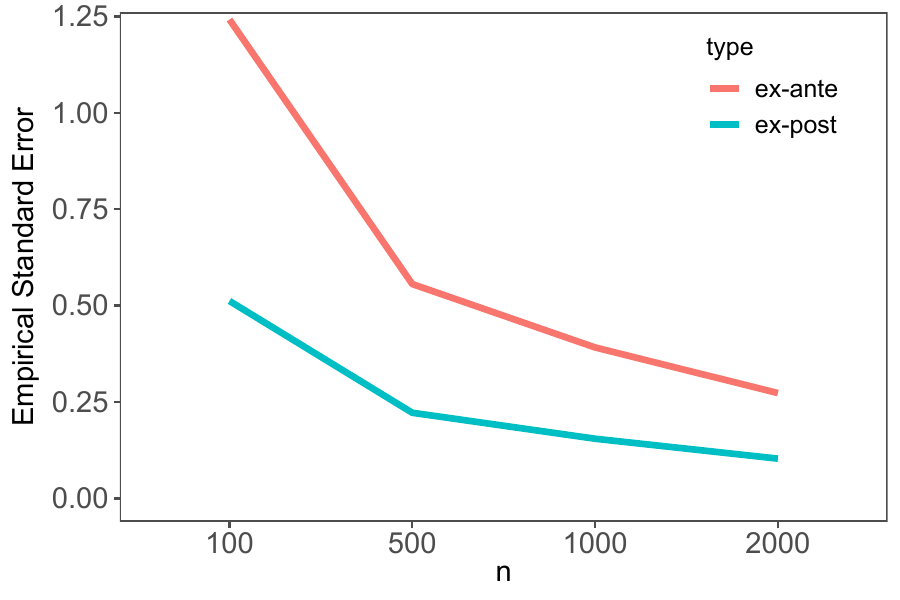}
	   \caption{Comparison of empirical standard error of the
             ex-ante and ex-post PAPE estimators (y-axis) for various
             sample sizes (x-axis).  }
	   \label{subfig:ante_post}
	 \end{subfigure}
	\caption{Numerical Experiments}
	\label{fig:neyman_exp}
\end{figure}

First, we investigate the effect of shifting potential outcomes by a
constant on the variance of estimators. Figure~\ref{subfig:var_shift}
plots the empirical standard deviation of the PAV estimator (the
vertical axis) with $n=100$ samples from the DGP as a function of
constant shift in potential outcomes (the horizontal axis). Here, we
centered the potential outcomes shift so that the optimal value of
$\delta$ given in Proposition~\ref{prop:invariance} is zero, i.e.,
\[\frac{n_0}{n}\kappa_{11}+\frac{n_1}{n}\kappa_{00}=0.\]
As predicted by our theoretical analysis, we find that balancing the
potential outcomes leads to a lower standard error in the estimator
due to the unbalanced nature of the relevant potential outcomes
$\mathbf{1}\{f(\bX_i)=t\}Y_i(t)$.

Second, we compare the statistical efficiency of the ex-ante and
ex-post PAPE estimators under the assumption $n_f=n_r=n/2$ and
$n_{r1}=n_{r0}=n/4$.  Consistent with our theoretical results,
Figure~\ref{subfig:ante_post} shows that the standard error of the
ex-ante estimator is consistently greater than that of the ex-post
estimator.  For example, when the sample size is 500, the former is
over twice the latter.

\section{Conclusion}

In this article, we provided a short overview of how Neyman's repeated
sampling framework can be utilized to experimentally evaluate the
performance of arbitrary ITRs.  We consider the two settings, one in
which an IRT is given and the other in which an ITR is estimated from
the same data.  We also demonstrated the new challenges that result
from the application of Neyman's framework, including the lack of
invariance of evaluation estimators and the need to incorporate the
uncertainty due to training of machine learning algorithms. We further
demonstrated how Neyman's repeated-sampling framework can highlight
the difference between the \emph{ex-ante} evaluation and
\emph{ex-post} evaluation of ITRs by showing that the \emph{ex-post}
evaluation is statistically more efficient.  Our ongoing work also
applies this framework to the estimation of heterogeneous treatment
effects discovered by machine learning algorithms
\citep{imai2023hetero}.  All together, we have shown that a century
after his original proposal, Neyman's analytical framework remains
relevant and widely applicable to the evaluation of today's causal
machine learning methods.

\paragraph*{Acknowledgments} The authors would like to thank Peng Ding and the two anonymous reviewers for their helpful and invaluable feedback during the review process. 
\paragraph*{Funding information}
None declared.
\paragraph*{Author contributions}
All authors have accepted responsibility for the entire content of this manuscript and approved its submission.
\paragraph*{Conflict of interest}
Authors state no conflict of interest.
\paragraph*{Data availability statement}
The numerical experiments included in the current study can be reproduced with the R scripts available at https://github.com/MichaelLLi/NeymanMLCode.

\bibliographystyle{ieeetr}
\bibliography{neymanml, imai, my}

\appendix
\section{Proof of Theorem~\ref{thm:PAPEaest}}
\label{app:PAPEaest}

We first consider the following intermediate estimator,
\begin{equation}
  \tilde\tau_f^{*}(\bZ_n) \ = \ \frac{1}{n_f}\sum_{i=1}^{n} Y_i(f(\bX_i))F_i-\frac{1}{n_r} \sum_{i=1}^n Y_i(T_i)(1-F_i). \label{eq:SAPEaest}
\end{equation}
This estimator differs from the \textit{ex-ante} estimator of the PAPE
$\hat\tau_f^*$, by a small factor, i.e.,
$\tilde\tau_f^{*} = (n-1)/n \hat\tau_f$ under the condition that $\hat{p}_f=n_{r1}/r_1$. The following lemma
derives the expectation and variance of this estimator.  Using this lemma,
the results of Theorem~\ref{thm:PAPEaest} can be obtained
immediately.
\begin{lemma}[Expectation and Variance of the Intermediate
  Estimator] \label{lemma:SAPEa} Under
  Assumptions~\ref{asm:SUTVA},~\ref{asm:randomsample},~and~\ref{asm:exante_random},
  the expectation and variance of the estimator given in
  Equation~\eqref{eq:SAPEaest} for estimating the PAPE defined in
  Equation~\eqref{eq:PAPE} are given by,
	\begin{eqnarray*}
		\E(\tilde\tau_f^{*}(\bZ_n))& = & \frac{n-1}{n}\tau_f, \\
		\V(\tilde\tau_f^{*}(\bZ_n)) & = &  \frac{\E(S_f^2)}{n_f}+\E\l\{\frac{\hat{p}_f^2S_1^2}{n_{r1}}+\frac{(1-\hat{p}_f)^2S_0^2}{n_{r0}}\r\}+\frac{1}{n^2} \l\{\tau_f^2 -n p_f(1-p_f)
		\tau^2 + 2(n-1)(2p_f-1) \tau_f\tau \r\}.
	\end{eqnarray*}
	%	where
	%	\begin{eqnarray*}
	%		S_f^2& = & \frac{1}{n-1} \sum_{i=1}^n (Y_i(f(\bX_i))-
	%		\overline{Y(f(\bX))})^2,\quad S_{t}^2  = \frac{1}{n-1} \sum_{i=1}^n (Y_i(t) -\overline{Y(t)})^2
	%	\end{eqnarray*}
\end{lemma}
\begin{proof}
  We first derive the bias expression. First, we take the expectation
  with respect to $T_i$,
  \begin{eqnarray*}
    & & \E[\tilde\tau_f^\ast(\bZ_n) \mid \{\bX_i, Y_i(1), Y_i(0), F_i\}_{i=1}^n] \\
    & = &\E\l[\frac{1}{n_f}\sum_{i=1}^{n}
          Y_i(f(\bX_i))F_i-\frac{1}{n_r}
          \sum_{i=1}^n \{Y_i(1)T_i+Y_i(0)(1-T_i)\}(1-F_i) \ \Bigl | \ \{\bX_i,
          Y_i(1), Y_i(0), F_i\}_{i=1}^n\r]\\
    &= &\frac{1}{n_f}\sum_{i=1}^{n}
         Y_i(f(\bX_i))F_i-\frac{1}{n_r}
         \sum_{i=1}^n \l\{Y_i(1)\frac{\sum_{i=1}^n
         f(\bX_i)}{n}+Y_i(0)\l(1-\frac{\sum_{i=1}^n
         f(\bX_i)}{n}\r)\r\}(1-F_i).
  \end{eqnarray*}
  Next, we take the expectation with respect to $F_i$:
  \begin{eqnarray*}
    & &\E\l[\frac{1}{n_f}\sum_{i=1}^{n}
        Y_i(f(\bX_i))F_i \r. \\
    & & \l. \hspace{.5in} -\frac{1}{n_r}
        \sum_{i=1}^n \l\{Y_i(1)\frac{\sum_{i=1}^n
        f(\bX_i)}{n}+Y_i(0)\l(1-\frac{\sum_{i=1}^n
        f(\bX_i)}{n}\r)\r\}(1-F_i) \ \Bigl | \ \{\bX_i, Y_i(1),
        Y_i(0)\}_{i=1}^n  \r]\\
    &= & \frac{1}{n} \sum_{i=1}^n Y_i(f(\bX_i))
         -\frac{1}{n_rn}\sum_{i=1}^{n}\sum_{j=1}^{n}\E\l[Y_i(1)f(\bX_j)(1-F_i)+Y_i(0)\l(1-f(\bX_j)\r)(1-F_i)\mid
         \{\bX_i, Y_i(1), Y_i(0)\}_{i=1}^n\r]\\
		& = & \frac{1}{n} \sum_{i=1}^n
                      Y_i(f(\bX_i)) -
                      \frac{1}{n_rn}
                      \sum_{i=1}^n \sum_{j=1}^n \l\{
                      \frac{n_rn}{n^2}Y_i(1)f(\bX_j)+Y_i(0)\frac{n_rn}{n^2}(1-f(\bX_j))\r\}\\
    &=&\frac{1}{n} \sum_{i=1}^n Y_i(f(\bX_i)) - \frac{1}{n^2}
        \sum_{i=1}^n \sum_{j=1}^n
        \l(Y_i(1)f(\bX_j)+Y_i(0)\l(1-f(\bX_j)\r)\r).
  \end{eqnarray*}
  Finally, we take the expectation over the sampling of
  $\{\bX_i, Y_i(1), Y_i(0)\}$:
  \begin{eqnarray*}
    & &\E\l[\frac{1}{n} \sum_{i=1}^n Y_i(f(\bX_i)) - \frac{1}{n^2}
        \sum_{i=1}^n \sum_{j=1}^n
        \l\{Y_i(1)f(\bX_j)+Y_i(0)\l(1-f(\bX_j)\r)\r\}\r]\\
    &= &
         \E\{Y_i(f(\bX_i))\}-p_f\E(Y_i(1))-\l(1-p_f\r)\E(Y_i(0))-\frac{1}{n^2}\sum_{i=1}^n
         \E\{\Cov(Y_i(1),f(\bX_i))+\Cov(Y_i(0),1-f(\bX_i))\}\\
    &=&\tau_f-\frac{1}{n}\Cov(Y_i(1)-Y_i(0),f(\bX_i))\\&=&\frac{n-1}{n}\tau_f.
  \end{eqnarray*}
  For the variance expression, we proceed as follows:
  \begin{align*}
    \V(\tilde\tau_f^{*}(\bZ_n)) & \ = \ \V\{\E(\tau_f^{*}(\bZ_n)\mid
                      \{\bX_i, Y_i(1), Y_i(0), F_i\}_{i=1}^n)\}+\E\{\V(\tau_f^{*}(\bZ_n)
                      \mid \bX_i, Y_i(1), Y_i(0), F_i)\}\\
                    & \ = \ \V\l[\frac{1}{n_f} \sum_{i=1}^n
                      Y_i(f(\bX_i))F_i  -
                      \frac{1}{n_r}
                      \sum_{i=1}^n\{Y_i(1)\hat{p}_f+Y_i(0)(1-\hat{p}_f)\}(1-F_i)\r]\\
                    & \hspace{.5in} +\E\l\{\frac{1}{n_r^2}\V\l[\sum_{i=1}^n
                      \{Y_iT_i+Y_i(1-T_i)\}(1-F_i) \ \Bigl  | \ \{\bX_i,
		Y_i(1), Y_i(0), F_i\}_{i=1}^n\r]\r\}.
  \end{align*}
  For the first term, we further use the law of total variance by
  conditioning on the sample, and center $F_i$ via the transformation
  $D_i = F_i - n_f/n$. For the second term, we use the results of
  \cite{neym:23}, with the following notation,
  \begin{equation*}
    S_{t}^2  \ = \  \frac{1}{n-1} \sum_{i=1}^n (Y_i(t)
    -\overline{Y(t)})^2, \quad
    S_{01} = \frac{1}{n-1} \sum_{i=1}^n(Y_i(0) - \overline{Y(0)}) (Y_i(1)
    - \overline{Y(1)}),
  \end{equation*}
  for $t=0,1$.  Then, the variance becomes,
  \begin{align*}
    \V(\tilde\tau_f^{*}(\bZ_n)) & \ = \ \E\l[\V\l\{\frac{1}{n}\sum_{i=1}^n D_i \l(\frac{n}{n_f}
                      Y(f(\bX_i)) +\frac{n}{n_r}
                      \widehat{Y}_i\r) \ \Bigl | \
                      \{\bX_i,Y_i(1),Y_i(0)\}_{i=1}^n \r\}\r] \\
                    & \quad +\V\l\{\frac{1}{n}\sum_{i=1}^n
                      \l(Y(f(\bX_i))-\widehat{Y}_i\r)\r\} +\E\l[\frac{1}{n_r}\l\{\frac{\hat{p}_f^2n_{r0}S_1^2}{n_{r1}}+\frac{(1-\hat{p}_f)^2n_{r1}S_0^2}{n_{r0}}-2\hat{p}_f(1-\hat{p}_f)S_{01}\r\}\r],
  \end{align*}
  where $\widehat{Y}_i=\hat{p}_fY_i(1)+(1-\hat{p}_f)Y_i(0)$.
  Then, we have,
  \begin{align*}
    \V(\tilde\tau_f^{*}(\bZ_n)) & \ = \
		\frac{\E(S_f^2)}{n_f}+\frac{\E(S_m^2)}{n_r}+\E\l[\frac{1}{n_r}\l\{\frac{\hat{p}_f^2n_{r0}S_1^2}{n_{r1}}+\frac{(1-\hat{p}_f)^2n_{r1}S_0^2}{n_{r0}}-2\hat{p}_f(1-\hat{p}_f)S_{01}\r\}\r]\\
		& \quad
		+\Cov((f(\bX_i)-\hat{p}_f)Y_i(1)-(f(\bX_i)-\hat{p}_f)Y_i(0),(f(\bX_i)-\hat{p}_f)Y_i(1)-(f(\bX_i)-\hat{p}_f)Y_i(0))\\
		& \ = \ \frac{\E(S_f^2)}{n_f}+\frac{\E(S_m^2)}{n_r}+\E\l[\frac{1}{n_r}\l\{\frac{\hat{p}_f^2n_{r0}S_1^2}{n_{r1}}+\frac{(1-\hat{p}_f)^2n_{r1}S_0^2}{n_{r0}}-2\hat{p}_f(1-\hat{p}_f)S_{01}\r\}\r]\\
		& \quad +\frac{1}{n^2}\l\{\tau_f^2 -n p_f(1-p_f)
		\tau^2 + 2(n-1)(2p_f-1) \tau_f\tau \r\} \\
		& \ = \ \frac{\E(S_f^2)}{n_f}+\E\l\{\frac{\hat{p}_f^2S_1^2}{n_{r1}}+\frac{(1-\hat{p}_f)^2S_0^2}{n_{r0}}\r\}+\frac{1}{n^2} \l\{\tau_f^2 -n p_f(1-p_f)
		\tau^2 + 2(n-1)(2p_f-1) \tau_f\tau \r\},
	\end{align*}
	where
	$S_m^2=\frac{1}{n-1}\sum_{i=1}^n
	(\widehat{Y}_i-\overline{\widehat{Y}})^2$ and the last equality
	follows from the
	$\E(S_m^2)=\E\{\hat{p}_f^2S_1^2+(1-\hat{p}_f^2)S_0^2+2\hat{p}_f(1-\hat{p}_f)S_{01}\}$.
\end{proof}
%\section{Difference of the PAV Variances}
%\label{app:var_diff_pav}
%
%To compute the difference of the variance, we begin by defining the
%following,
%\begin{eqnarray*}
%	A_i & = & f(\bX_i)Y_i(1)-\overline{f(\bX)Y(1)},\quad B_i\ = \ (1-f(\bX_i))Y_i(0)-\overline{(1-f(\bX))Y(0)}
%\end{eqnarray*}
%Then, a simple algebraic manipulation yields,
%\begin{eqnarray*}
%	\V(\hat{\lambda}_f) & = & \E\l(\sum_{i=1}^n \frac{A_i^2}{n_1(n-1)}+\frac{B_i^2}{n_0(n-1)}\r)\\
%	\V(\hat{\lambda}_f^*) & = & \E\l(\sum_{i=1}^n \frac{A_i^2+B_i^2+2A_iB_i}{n_f(n-1)} + \frac{(Y_i(0)-\overline{Y(0)})^2}{n_r(n-1)}\r).
%\end{eqnarray*}
%Under the simplifying assumptions of $n_1=n_0=n_f=n_r$, we have:
%\begin{align*}
%	&\V(\hat{\lambda}_f^*) - \V(\hat{\lambda}_f) 
%	\\&= 2\E\l(\sum_{i=1}^n\frac{2A_iB_i}{n(n-1)} + \frac{(Y_i(0)-\overline{Y(0)})^2}{n(n-1)} \r)\\&= 2\frac{\V(Y_i(0))}{n} + 4\Cov(f(\bX_i)Y_i(1),(1-f(\bX_i))Y_i(0))\\&=\frac{2}{n} \left(\frac{\V(Y_i(0))}{n} - 2p_f(1-p_f)\E(Y_i(0)\mid f(\bX_i)=0)\E(Y_i(1)\mid f(\bX_i)=1)\right)
%\end{align*}
%where the final equality utilizes the derivation
%\begin{eqnarray*}
%	& &\Cov(f(\bX_i)Y_i(1),\left(1-f(\bX_i)\right)Y_i(0))\\
%	&= &\E\{f(\bX_i)Y_i(1)\left(f(\bX_i)-1\right)Y_i(0)\}-\E\{f(\bX_i)Y_i(1)\}\E\{\left(1-f(\bX_i)\right)Y_i(0)\}\\
%	&=&-\Pr(f(\bX_i)=1)\E(Y_i(1)\mid f(\bX_i)=1)\Pr(f(\bX_i)=0)
%	\E(Y_i(0)\mid f(\bX_i)=0)\\
%	&=&-p_f(1-p_f)\E(Y_i(0)\mid f(\bX_i)=0)\E(Y_i(1)\mid f(\bX_i)=1)
            %  \end{eqnarray*}
\section{Proof of Proposition~\ref{prop:invariance}}
\label{app:invariance}
By definition, we have that:
\begin{align*}
& \V(\hat\lambda^\delta_f(\bZ))-\V(\hat\lambda_f(\bZ)) \\
\ = \ &\V\left(\hat\lambda_f(\bZ) + \frac{\delta}{n_1} \sum_{i=1}^n T_if(\bX_i)+ \frac{\delta}{n_0} \sum_{i=1}^n (1-T_i)(1-f(\bX_i))\right) - \V(\hat\lambda_f(\bZ))
  \\
  \ =  \ &2\delta \cdot \Cov\left(\hat\lambda_f(\bZ) , \frac{1}{n_1}
           \sum_{i=1}^n T_if(\bX_i)+ \frac{1}{n_0} \sum_{i=1}^n
           (1-T_i)(1-f(\bX_i))\right) \\
  & \hspace{.25in} + \delta^2 \cdot \V\left(\frac{1}{n_1} \sum_{i=1}^n T_if(\bX_i)+\frac{1}{n_0} \sum_{i=1}^n (1-T_i)(1-f(\bX_i))\right)
\\ = \ &2\delta \left\{\frac{\Cov(f(\bX_i)Y_i(1),f(\bX_i))}{n_1} +
         \frac{\Cov((1-f(\bX_i))Y_i(1),1-f(\bX_i))}{n_0}\right\} + \delta^2 \left\{\frac{p_f(1-p_f)}{n_1}+ \frac{p_f(1-p_f)}{n_0}\right\}
\\ = \ &\delta p_f(1-p_f)\left(\frac{2}{n_1}\E[Y_i(1) \mid
         f(\bX_i)=1]+ \frac{2}{n_0}\E[Y_i(0) \mid f(\bX_i)=0]+
         \delta \cdot \frac{n}{n_1n_0}\right).
\end{align*}
Define $b =\frac{n}{n-1}=1+\frac{1}{n-1}$.  Then, we have:
\begin{align*}
  &\V(\hat\tau^\delta_f(\bZ))-\V(\hat\tau_f(\bZ))
  \\= \ &\V\left(\hat\tau_f(\bZ) + \frac{b\delta}{n_1} \sum_{i=1}^n T_if(\bX_i)+ \frac{b\delta}{n_0} \sum_{i=1}^n (1-T_i)(1-f(\bX_i))\right) - \V(\hat\tau_f(\bZ))
  \\= \ &2b\delta \cdot \Cov\left(\hat\tau_f(\bZ) , \frac{1}{n_1}
          \sum_{i=1}^n T_if(\bX_i)+ \frac{1}{n_0} \sum_{i=1}^n
          (1-T_i)(1-f(\bX_i))\right) + b^2\delta^2 \cdot \V\left(\frac{1}{n_1} \sum_{i=1}^n T_if(\bX_i)+\frac{1}{n_0} \sum_{i=1}^n (1-T_i)(1-f(\bX_i))\right)
  \\= \ &2b^2\delta \cdot \Cov\left(\hat\lambda_f(\bZ) , \frac{1}{n_1}
       \sum_{i=1}^n T_if(\bX_i)+ \frac{1}{n_0} \sum_{i=1}^n
       (1-T_i)(1-f(\bX_i))\right) \\
  & \hspace{.25in} - 2 b^2\delta \cdot
    \Cov\left(\frac{\hat{p}_f}{n_1} \sum_{i=1}^n Y_iT_i+ \frac{1-\hat{p}_f}{n_0}  \sum_{i=1}^n Y_i(1-T_i) ,
    \frac{1}{n_1} \sum_{i=1}^n  T_if(\bX_i)+ \frac{1}{n_0}
    \sum_{i=1}^n (1-T_i)(1-f(\bX_i))\right)\\
  & \hspace{.25in} + b^2\delta^2\cdot \V\left(\frac{1}{n_1} \sum_{i=1}^n T_if(\bX_i)+\frac{1}{n_0} \sum_{i=1}^n (1-T_i)(1-f(\bX_i))\right)
  \\= \ & b^2 \delta p_f(1-p_f)\left(\frac{2}{n_1}\E[Y_i(1) \mid
          f(\bX_i)=1]+ \frac{2}{n_0}\E[Y_i(0) \mid f(\bX_i)=0]+ \delta
          \cdot \frac{n}{n_1n_0}\right)\\
  & \hspace{.25in} -  2b^2\delta \left(\frac{n\E[(\hat{p}_f Y_i(1) - \hat{p}_f
    \overline{Y_i(1)})(f(\bX_i)-\hat{p}_f)]}{n_1(n-1)} \right. \\
  & \hspace{.75in} \left. + \frac{n\E[\{(1-\hat{p}_f) Y_i(0) - (1-\hat{p}_f) \overline{Y_i(0)}\}\{1-f(\bX_i)-(1-\hat{p}_f)\}]}{n_0(n-1)}\right)
  \\= \ & b^2 \delta p_f(1-p_f)\left(\frac{2}{n_1}\E[Y_i(1) \mid
          f(\bX_i)=1]+ \frac{2}{n_0}\E[Y_i(0) \mid f(\bX_i)=0]+ \delta
          \cdot \frac{n}{n_1n_0}\right)\\ & \hspace{.25in} - 2 b^2\delta \left\{\frac{p_f+p_f^2(n-2)}{n_1n^2}(\kappa_{11}-\E[Y_i(1)])+ \frac{(1-p_f)+(1-p_f)^2(n-2)}{n_0n^2}(\kappa_{00}-\E[Y_i(0)])\right\}
  \\= & b^2 \delta p_f(1-p_f)\left(\frac{2}{n_1}\E[Y_i(1) \mid
        f(\bX_i)=1]+ \frac{2}{n_0}\E[Y_i(0) \mid f(\bX_i)=0]+ \delta
        \cdot \frac{n}{n_1n_0}\right)+O\left(\frac{\delta}{n^2}\right).
\end{align*}
\qed

\section{Difference of the PAPE Variances}
\label{app:var_diff_pape}

To compute the difference of the two PAPE variances, we first define
the following,
\begin{eqnarray*}
	A_i & = & \hat{p}_fY_i(1)-\hat{p}_f\overline{Y(1)},\quad B_i \ = \ (1-\hat{p}_f)Y_i(0)-(1-\hat{p}_f)\overline{Y(0)},\\
	C_i & = & f(\bX_i)Y_i(1)-\overline{f(\bX)Y(1)},\quad D_i\ = \ (1-f(\bX_i))Y_i(0)-\overline{(1-f(\bX))Y(0)}.
\end{eqnarray*}
Then, a simple algebraic manipulation yields,
\begin{eqnarray*}
	\V(\hat{\tau}_f(\bZ_n)) & = & \frac{n^2}{(n-1)^2}\E\l(\sum_{i=1}^n \frac{A_i^2+C_i^2-2A_iC_i}{n_1(n-1)}+\frac{B_i^2+D_i^2-2B_iD_i}{n_0(n-1)} +\xi\r),\\
	\V(\hat{\tau}_f^*(\bZ_n)) & = & \frac{n^2}{(n-1)^2}\E\l(\sum_{i=1}^n \frac{A_i^2}{n_{r1}(n-1)}+\frac{B_i^2}{n_{r0}(n-1)}+\frac{C_i^2+D_i^2+2C_iD_i}{n_f(n-1)} +\xi\r),
\end{eqnarray*}
where
$\xi = \frac{1}{n^2} \l\{\tau_f^2 -n p_f(1-p_f) \tau^2 +
2(n-1)(2p_f-1) \tau_f\tau \r\}$.  Given these expressions, the
difference is given by,
\begin{eqnarray*}
	\V(\hat{\tau}_f^*)-\V(\hat{\tau}_f)&=&\frac{n^2}{(n-1)^2}\E\l(\sum_{i=1}^n \frac{A_i^2(n_1-n_{r1})}{n_{r1}n_1(n-1)}+\frac{B_i^2(n_0-n_{r0})}{n_{r0}n_0(n-1)}+\frac{C_i^2(n_1-n_f)}{n_fn_1(n-1)}+\frac{D_i^2(n_0-n_f)}{n_fn_0(n-1)}\r.\\
	& & +\biggl. \frac{2C_iD_i}{n_f(n-1)}+\frac{2A_iC_i}{n_1(n-1)}+\frac{2B_iD_i}{n_0(n-1)} \biggr).
\end{eqnarray*}
Under the assumption that $n_1=n_0=n_f=n_r=n/2$ and
$n_{r0}=n_{r1}=n/4$, we have,
\begin{eqnarray*}
	\V(\hat{\tau}_f^*(\bZ_n))-\V(\hat{\tau}_f(\bZ_n))&=&\frac{2n}{(n-1)^2}\E\l(\sum_{i=1}^n \frac{A_i^2+B_i^2}{n-1}+\frac{2C_iD_i+2A_iC_i+2B_iD_i}{n-1} \r)
	\\& = & \frac{2n}{(n-1)^2}\l[\E\l\{p_f^2S_1^2+(1-p_f)^2S_0^2\r\}
	+2\Cov(f(\bX_i)Y_i(1),(1-f(\bX_i))Y_i(0)) \r. \\
	& & \l. + \biggl.2p_f\Cov(f(\bX_i)Y_i(1),Y_i(1)) + 2(1-p_f)\Cov((1-f(\bX_i))Y_i(0),Y_i(0))\r.].
\end{eqnarray*}
Finally, note the following,
\begin{eqnarray*}
	& &\Cov(f(\bX_i)Y_i(1),\left(1-f(\bX_i)\right)Y_i(0))\\
	&= &\E\{f(\bX_i)Y_i(1)\left(f(\bX_i)-1\right)Y_i(0)\}-\E\{f(\bX_i)Y_i(1)\}\E\{\left(1-f(\bX_i)\right)Y_i(0)\}\\
	&=&-\Pr(f(\bX_i)=1)\E(Y_i(1)\mid f(\bX_i)=1)\Pr(f(\bX_i)=0)
	\E(Y_i(0)\mid f(\bX_i)=0)\\
	&=&-p_f(1-p_f)\E(Y_i(0)\mid f(\bX_i)=0)\E(Y_i(1)\mid f(\bX_i)=1),
\end{eqnarray*}
and
\begin{eqnarray*}
	p_f\Cov (f(\bX_i)Y_i(1),Y_i(1))
	&=& p_f^2 \l\{\E(Y_i^2(1) \mid f(\bX_i) = 1) -
	\E(Y_i(1))\E(Y_i(1) \mid f(\bX_i) =1)\r\} \\
	(1-p_f)\Cov (f(\bX_i)Y_i(0),Y_i(0))
	&=& (1-p_f)^2 \l\{\E(Y_i^2(0) \mid f(\bX_i) = 0) -
	\E(Y_i(0))\E(Y_i(0) \mid f(\bX_i) =0)\r\}.
\end{eqnarray*}
Hence, we have,
\begin{eqnarray*}
	& &\V(\hat{\tau}_f^*(\bZ_n))-\V(\hat{\tau}_f(\bZ_n)) \\
	& = & \frac{2n}{(n-1)^2}\l[p_f^2\V(Y_i(1))+(1-p_f)^2\V(Y_i(0))
	-2p_f(1-p_f)\E(Y_i(0) \mid f(\bX_i)=0)\E(Y_i(1)\mid f(\bX_i)=1) \r. \nonumber\\
	& & \hspace{.5in} +  2 p_f^2 \l\{\E(Y_i^2(1) \mid f(\bX_i) = 1) -
	\E(Y_i(1))\E(Y_i(1) \mid f(\bX_i) =1)\r\}\\
	& & \l. \hspace{.5in} +  2(1-p_f)^2 \l\{\E(Y_i^2(0) \mid f(\bX_i) = 0) -
	\E(Y_i(0))\E(Y_i(0) \mid f(\bX_i) =0)\r\} \r] \nonumber\\
	& = & \frac{2n}{(n-1)^2}\l[p_f^2\V(Y_i(1))+(1-p_f)^2\V(Y_i(0))
	-2p_f(1-p_f)\E(Y_i(0) \mid f(\bX_i)=0)\E(Y_i(1)\mid f(\bX_i)=1) \r. \nonumber\\
	& & \hspace{.5in} +  2 p_f^2 \l\{\V(Y_i(1) \mid f(\bX_i) = 1)
	\r.\\ & & +\l. (1-p_f)\{\E(Y_i(1)\mid f(\bX_i)=1)
	-\E(Y_i(1)\mid f(\bX_i)=0)\} \E(Y_i(1) \mid f(\bX_i) =1)\r\} \\
	& & \hspace{.5in} +  2 (1-p_f)^2 \l\{\V(Y_i(0) \mid f(\bX_i) = 0)
	\r.\\ & & \hspace{.75in} +\l. p_f\{\E(Y_i(0)\mid f(\bX_i)=0)
	-\E(Y_i(0)\mid f(\bX_i)=1)\} \E[Y_i(0) \mid f(\bX_i) =0]\r\}.
\end{eqnarray*}
\qed

\section{Comparison under the Simplifying Assumptions}
\label{app:comp_simple}

Define $M_{st} = \E(Y_i(s) \mid f(\bX_i) =t)$ for $s,t \in
\{0,1\}$. Then, we can rewrite the variance difference as,
\begin{eqnarray*}
	\V(\hat{\tau}_f^*(\bZ_n))-\V(\hat{\tau}_f(\bZ_n))
	& = & \frac{2n}{(n-1)^2}\l[p_f^2\V(Y_i(1))+(1-p_f)^2\V(Y_i(0))
	-2p_f(1-p_f)M_{11}M_{00}  \r. \\ &  & \hspace{.5in}+  2 p_f^2 \l\{\V(Y_i(1) \mid f(\bX_i) = 1) + (1-p_f)(M_{11}-M_{10}) M_{11}\r\} \\
	& & \hspace{.5in} +  \l. 2 (1-p_f)^2 \l\{\V(Y_i(0) \mid f(\bX_i) = 0)+ p_f(M_{00}-M_{01}) M_{00}\r\} \r].
\end{eqnarray*}
Now, consider a constant shift of the outcome variable, i.e.,
$Y_i(t)+\delta $ for $t=0,1$. Then, the variance difference becomes,
\begin{eqnarray*}
	& &\V(\hat{\tau}_f^*(\bZ_n))-\V(\hat{\tau}_f(\bZ_n)) \\
	& = & \frac{2n}{(n-1)^2}\l[p_f^2\V(Y_i(1))+(1-p_f)^2\V(Y_i(0))
	-2p_f(1-p_f)(M_{11}+\delta)(M_{00}+\delta)  \r. \\ &  & \hspace{.5in}+  2 p_f^2 \l\{\V(Y_i(1) \mid f(\bX_i) = 1) + (1-p_f)(M_{11}-M_{10}) (M_{11}+\delta)\r\} \\
	& & \hspace{.5in} +  \l. 2 (1-p_f)^2 \l\{\V(Y_i(0) \mid f(\bX_i) = 0)+ p_f(M_{00}-M_{01}) (M_{00}+\delta)\r\} \r]\\
	& = & \frac{2n}{(n-1)^2}\l[p_f^2\V(Y_i(1))+(1-p_f)^2\V(Y_i(0))
	-2p_f(1-p_f)M_{11}M_{00}  \r. \\ &  & \hspace{.5in}+  2 p_f^2 \l\{\V(Y_i(1) \mid f(\bX_i) = 1) + (1-p_f)(M_{11}-M_{10}) M_{11}\r\} \\
	& & \hspace{.5in} +   2 (1-p_f)^2 \l\{\V(Y_i(0) \mid f(\bX_i) = 0)+ p_f(M_{00}-M_{01}) M_{00}\r\} \\
	& & \hspace{.5in} -   2 p_f(1-p_f) \delta^2  + \l. 2p_f(1-p_f)\delta \l\{p_f(M_{11}-M_{10})+(1-p_f)(M_{00}-M_{01})-M_{11}-M_{00}\r\}\r]\\
	& = & \frac{2n}{(n-1)^2}\l[p_f^2\V(Y_i(1))+(1-p_f)^2\V(Y_i(0))
	-2p_f(1-p_f)M_{11}M_{00}  \r. \\ &  & \hspace{.5in}+  2 p_f^2 \l\{\V(Y_i(1) \mid f(\bX_i) = 1) + (1-p_f)(M_{11}-M_{10}) M_{11}\r\} \\
	& & \hspace{.5in} +   2 (1-p_f)^2 \l\{\V(Y_i(0) \mid f(\bX_i) = 0)+ p_f(M_{00}-M_{01}) M_{00}\r\} \\
	& & \hspace{.5in} -   2 p_f(1-p_f) \delta^2  - \l.2p_f(1-p_f)\delta \l\{p_f(M_{00}+M_{10})+(1-p_f)(M_{11}+M_{01})\r\}\r].
\end{eqnarray*}
Thus, we observe that the variance difference decreases by,
\begin{equation*}
	2 p_f(1-p_f) \delta^2  + 2p_f(1-p_f)\delta
	\l\{p_f(M_{00}+M_{10})+(1-p_f)(M_{11}+M_{01})\r\}.
\end{equation*}
Since the \textit{ex-ante} estimator is completely unaffected by this
change, the constant shift increases the variance of the
\textit{ex-post} evaluation estimator by the same amount.  Under the
simplifying assumptions, we have,
\begin{equation*}
	M_{11}+M_{01} \ = \ M_{00}+M_{10} \ = \  0.
\end{equation*}
Therefore, we can bound the difference in variance from below as follows,
\begin{eqnarray*}
	& & \V(\hat{\tau}_f^\ast(\bZ_n))-\V(\hat{\tau}_f(\bZ_n))  \\
	& = & \frac{2n}{(n-1)^2}\l[p_f^2\V(Y_i(1))+(1-p_f)^2\V(Y_i(0))
	-2p_f(1-p_f)M_{11}M_{00}  \r. \\ &  & \hspace{.5in}+  2 p_f^2 \l\{\V(Y_i(1) \mid f(\bX_i) = 1) + (1-p_f)(M_{11}-M_{10}) M_{11}\r\} \\
	& & \hspace{.5in} +  \l. 2 (1-p_f)^2 \l\{\V(Y_i(0) \mid f(\bX_i) = 0)+ p_f(M_{00}-M_{01}) M_{00}\r\} \r]\\
	& = & \frac{2n}{(n-1)^2}\l[p_f^2\V(Y_i(1))+(1-p_f)^2\V(Y_i(0))
	-2p_f(1-p_f)M_{11}M_{00}  \r. \\ &  & \hspace{.5in}+  2 p_f^2 \l\{\V(Y_i(1) \mid f(\bX_i) = 1) + (1-p_f)(M_{11}+M_{00}) M_{11}\r\} \\
	& & \hspace{.5in} +  \l. 2 (1-p_f)^2 \l\{\V(Y_i(0) \mid f(\bX_i) = 0)+ p_f(M_{00}+M_{11}) M_{00}\r\} \r]\\
	& = & \frac{2n}{(n-1)^2}\l[p_f^2\V(Y_i(1))+(1-p_f)^2\V(Y_i(0))
	+2p_f^2\V(Y_i(1) \mid f(\bX_i)=1) + 2(1-p_f)^2\V(Y_i(0) \mid f(\bX_i)=0) \r.\\
	& & \l. +2p_f(1-p_f)\l[(1-p_f)M_{00}^2+p_fM_{11}^2\r] \r] \\
	%	& = & \frac{2n}{(n-1)^2}\l[p_f^2\V(Y_i(1))+(1-p_f)^2\V(Y_i(0))
	%	+2p_f^2\V(Y_i(1) \mid f(\bX_i)=1) + 2(1-p_f)^2\V(Y_i(0) \mid f(\bX_i)=0) \r.\\
	%	& & \l. +2p_f(1-p_f)\l[(1-p_f)\{\E(Y_i(0)\mid f(\bX_i)=0)\}^2+p_f\{\E(Y_i(1)\mid f(\bX_i)=1)\}^2\r] \r] \\
	& \geq & 0.
\end{eqnarray*}
\qed

\section{The Outcome Model for the Numerical Study}
\label{app:sim}

\begin{align*}
  \E(Y_i(t) \mid \bX_i) &= 1.60+ 0.53\times
    x_{29}-3.80\times x_{29}(x_{29}-0.98)(x_{29}+0.86) -0.32 \times
    \bone\{x_{17}>0\}\\& + 0.21 \times \bone\{x_{42}>0\}-0.63 \times
    x_{27}+4.68 \times \bone\{x_{27}<-0.61\}-0.39 \times (x_{27}+0.91)
    \bone\{x_{27}<-0.91\}\\&+ 0.75 \times \bone\{x_{30}\leq0\}-1.22
    \times \bone\{x_{54}\leq0\}+0.11 \times x_{37}
    \bone\{x_{4}\leq0\}-0.71 \times \bone\{x_{17}\leq0, t=0\}\\&-1.82
    \times \bone\{x_{42}\leq 0,t=1\}+0.28 \times \bone\{x_{30}\leq
    0,t=0\}\\&+\{0.58\times x_{29}-9.42 \times
    x_{29}(x_{29}-0.67)(x_{29}+0.34)\}\times\bone\{t=1\}\\&+(0.44
    \times x_{27}-4.87\times
    \bone\{x_{27}<-0.80\})\times\bone\{t=0\}-2.54 \times \bone\{t=0,
    x_{54}\leq 0\}.
\end{align*}

\end{document}